\documentclass[final]{siamltex}

\usepackage{amsmath}
\usepackage{amssymb}
\usepackage{graphicx}
\usepackage{subfigure}
\usepackage{psfig}
\usepackage[ruled,vlined]{algorithm2e}

\graphicspath{{fig/}}

\newcommand{\univec}{\mathop{\vec{\mathbf{1}}}}
\newtheorem{problem}{Problem}[section]
\newtheorem{assumption}{Assumption}[section]
\newtheorem{example}{Example}[section]
\newtheorem{conjecture}{Conjecture}[section]
\newtheorem{remark}{Remark}[section]
\title{Parameter Optimization of Multi-Agent Formations based on LQR Design\thanks{This work is
completed during H. Huang's visit at the Australian National
University, which is supported by The National Natural Science
Foundation of China under grant 61074031; C. Yu is supported by the
Australian Research Council through an APD Fellowship under
DP-0877562 and subsequently a QEII Fellowship under DP-110100538}.}

\author{Huang Huang\thanks{School of Automation, Beijing Institute of
Technology, Beijing 100081, China ({\tt huanghuang@bit.edu.cn}).}
        \and Changbin Yu\thanks{RSISE, Building 115, The Australian
National University, Canberra, ACT, 0200, Australia ({\tt
brad.yu@anu.edu.au}).}}

\begin{document}

\maketitle

\begin{abstract}
In this paper we study the optimal formation control of multiple
agents whose interaction parameters are adjusted upon a cost
function consisting of both the control energy and the geometrical
performance. By optimizing the interaction parameters and by the
linear quadratic regulation(LQR) controllers, the upper bound of the
cost function is minimized. For systems with homogeneous agents
interconnected over sparse graphs, distributed controllers are
proposed that inherit the same underlying graph as the one among
agents. For the more general case, a relaxed optimization problem is
considered so as to eliminate the nonlinear constraints. Using the
subgradient method, interaction parameters among agents are
optimized under the constraint of a sparse graph, and the optimum of
the cost function is a better result than the one when agents
interacted only through the control channel. Numerical examples are
provided to validate the effectiveness of the method and to
illustrate the geometrical performance of the system.
\end{abstract}

\begin{keywords}
formation control, LQR, parameter optimization
\end{keywords}

\begin{AMS}
15A15, 15A09, 15A23
\end{AMS}

\pagestyle{myheadings} \thispagestyle{plain} \markboth{H. HUANG AND
C. YU}{SIAM MACRO EXAMPLES}

\section{Introduction}
The study of formations for a group of agents is inspired by the
behaviors of various animal species in nature. For instance, fish,
birds and ants always work in a cooperative manner so as to
accomplish tasks that are beyond the capability of single ones. The
formation of agents has its wide range of applications both in
civilian and military life, where there are generally three basic
parts: the determination of the underlying graph, the development of
cooperative algorithms among agents that concerns with an exact
assignment and the deployment of controllers that are responsible
for the stability of the overall system.

Upon the construction of the underlying communication graphs,
\cite{Anderson08rigidMagazine,Yu09minimalPersistent} discovered that
a geometry with an underlying graph being a rigid one or a
persistent one would not transform under smooth movements once
settled. Meanwhile, other graph properties such as connectivity,
strong connectivity or even the existence of a spanning tree are
proved to be sufficient to the convergence of consensus algorithms
\cite{Jadbabaie03nearestNeighbor,Ren08distributed,Fax04,Xiao09finiteFormation}.

In general, the optimization objectives in a formation system
include costs of either the relative formation
errors\cite{Liu10pathplanning4error}, the absolute position errors
\cite{Raffard04optimizationFormation,Semsar08optimalFormation} or
the minimal traveling
distances\cite{Zhang07optimalFormationConstrained}. On the other
hand, \cite{Bhatt09geometric} discussed the optimal relative layout
of wheeled robots in terms that the kinetic energy
 is minimized. \cite{Borrelli08LQR} focused on
distributed control design of large-scale dynamically isolated
systems using LQR optimization. In the trajectory planning of
formation reconfiguration, the main concerns are  the minimal
executing time\cite{Tillerson02}, the fuel
consumption\cite{Kim09pathplanning} and the kinetic energy
expenditure\cite{Afflitto10}. However, nearly all of the studies
made the assumptions that the local dynamics of agents were
pre-specified ones and agents were only coupled through the control
channels rather than being dynamically related in the open loop. In
real time applications such as sensor networks, it is always
flexible for us to control the behavior of each individual agent by
designating its interaction rules rather than simply setting to
zero. Thus problem arises on how to utilize this flexibility on
interaction parameters so as to improve the performance of a
formation system? Here by performance we focus on the geometries'
deformations with respect to the desired one during convergence and
the control energy cost.

Although the problem of parameter optimization is brand new in
formation control, in recent years, a similar problem has been
discussed in the field of consensus algorithms as seen in a limited
number of literatures. To our best knowledge, the optimal weight
design over a fixed network topology is addressed in
\cite{XiaoBoyd04fastLinearIterations} and \cite{Xiao07weight}.
  The optimal weight design under
random networks and switching networks are explored in
\cite{Jakovetic10weightOptimization} and \cite{Jakovetic10switching}
respectively.

In conventional formation control that concerns formation
maintenance while moving towards a destination, the discussion is
divided into two subsequential steps with the first one being the
relative formation maintenance or formation stabilization
irrespective of the destination.  For the second step, all agents in
formation move towards the destination led by one or two agents
known as the formation leader. This process is always termed
destination attainment. During the process, for the purpose of
formation maintenance, the translational speed of the formation
system is bounded due to the physical restrictions of each agent. In
bearing-only sensor-target localization, when a fleet of UAVs aiming
to locate a moving object accurately and as fast as possible, they
are expected to maintain an optimal geometry\cite{Bishop10sensor}
during the entire process. In such a case, formation maintenance is
more important. On the contrary, in payload transportation,
formation maintenance is not strictly demanded while a rapid task
execution time is the major concern.

To this end, in this paper we study the optimal formation control by
taking both the above two indicators into account simultaneously.
With different ratios of attentions being paid to those two
indicators, we define three kinds of cooperative performance  and
characterize each of them using a quadratic cost matrix. In order to
find the optimal interaction parameters between agents such that the
upper bound of the cost function is minimized, the LQR control
strategy is considered. Although in some literatures distributed
optimization was discussed over dynamically isolated agents, i.e.,
agents were interconnected only through the control
channel\cite{Langbort09minimal,Borrelli08LQR}, we aim to discover
that the optimum in those literatures could be further minimized
through some carefully selected interaction parameters among agents.

The remainder of this paper is organized as follows: In Section II,
we introduce the basic theories on LQR control into the optimal
formation control. In Section III, a simple system with two agents
is considered to illustrate the relationship between different kinds
of cooperative
 performance and the cost function we propose. In Section IV, the
optimization problem with multiple agents is investigated. For the
special case with homogeneous agents, we propose distributed
controllers design method by considering dynamically isolated
agents. For the more general case with heterogeneous agents, in
order to avoid the nonlinearity in the constraint conditions, a
relaxed optimization problem is further proposed and the structured
parameter matrix is calculated based on the subgradient method. In
Section V simulations are presented to validate those theoretical
results and to demonstrate the relationships between the cost
function and the cooperative performance. Finally conclusions are
given in Section VI.

\section{Preliminary}
Let $\mathbb{R}_S$ denote the set of real symmetric matrices,
$\mathbb{R}^+_S$ the set of positive definite matrices and
$\mathbb{R}_\Lambda$ the set of $n\times n$ diagonal matrices. The
set of eigenvalues of matrix $A$ is denoted by $\lambda(A)$ where
the largest one is $\bar{\lambda}(A)$.
The $n\times n$ identity matrix is $I_n$ whose $i$th column is
denoted by $e_i$. $\mathbb{M}$ is the set of nonsingular matrices.
Let $A_{i,j}$ be the $ij$th entry, then the $ij$th block of a matrix
$A$ and the $i$th block of a vector $v$ are defined by
$$A_{\bar{ij}}=\begin{bmatrix}
  A_{2i-1,2j-1} & A_{2i-1,2j}\\
A_{2i,2j-1} & A_{2i,2j}
\end{bmatrix}, v_{\bar{i}}=\begin{bmatrix}
v_{2i-1} \\ v_{2i}
\end{bmatrix}$$
respectively. We use $\diag(v), v\in\mathbb{R}^n$ or $\diag([v_i])$
to denote a diagonal matrix with diagonal entries from vector $v$.
Similarly, $\diag(A)$ when $A\in\mathbb{R}^{n\times n}$ is a vector
with entries from the diagonal of $A$. $\Lambda(A)$ is a diagonal
matrix that has the same diagonal entries as matrix $A$. We have
$\Lambda(A)=\diag(\diag(A))$. The 2-block diagonal of matrix $A$ is
denoted by $\bar{\Lambda}(A)$ where
$\bar{\Lambda}(A)_{\bar{ii}}=A_{\bar{ii}}$ and all the other blocks
are zeros.

An undirected  graph is denoted by $G=(V,E)$ with $|V|$ vertices and
$|E|$ edges. Node $i$ is a neighbor of node $j$, denoted by $i\sim
j$, if they are connected by an edge in $E$. We do not consider
selfloops in a graph. The adjacency matrix of $G$ is
$\mathcal{A}(G)\in\mathbb{R}_S^{|V|\times |V|}$ and is determined by
$$\mathcal{A}(G)_{ij}=\begin{cases}
  1,& i\sim j\\
  0,&\text{otherwise}
\end{cases}$$
A Graph $G$ is uniquely determined by an adjacency matrix thus $G$
is also said to be \emph{generated} from $\mathcal{A}$.

The incidence matrix of $G$ is the $|E|\times |V|$ matrix $H$
defined by $H_{ij}=1$ if vertex $j$ is an endpoint of edge $i$ and
$H_{ij}=0$ otherwise. The oriented incidence matrix $\bar{H}$ is
defined by replacing a 1 in each row by a $-1$.

For a symmetric matrix $A$ and an adjacency matrix $\mathcal{A}$
such that $A=A\circ\mathcal{A}$ where $\circ$ is the Hadamard
product, the \emph{underlying graph} of matrix $A$ is the one that
is generated from $\mathcal{A}$.

Consider a  system modeled by the standard state-space equation
$\dot{x}=A x+B u$ where $A,B$ is controllable and the quadratic cost
function
\begin{equation}\label{eq:J}
  J(x(0),u):\ =\int_0^\infty x^T(t) Q x(t)+u^T(t) u(t) dt
\end{equation}
where $Q\in \mathbb{R}^+_S$ and $x(0)$ is the initial state.
$\int_0^\infty x^T(t) Q x(t)$ is called the energy expenditure of
$x$. Based on LQR control theory, if we consider the state-feedback
controller
\begin{equation}
  u(t)=-B^T X^+x(t)
\end{equation}
where $X^+$ is the unique positive definite solution to the
algebraic Riccati equation
\begin{equation}\label{eq:riccati}
 A X+X A^T-XB B^T X+Q=0
\end{equation}
then the cost function $J$ in \eqref{eq:J} achieves its minimum of
\begin{equation}
  \bar{J}(x(0))=x(0)^T X^+ x(0)
\end{equation}
In order to concentrate on the parameter design issue, throughout
the paper it is assumed that $\|x(0)\|_2=1$. For the worst case,
$\bar{J}$ is upper bound by $\bar{J}=\bar{\lambda}(X^+)$, and
minimizing this upper bound of $\bar{J}$ is the main concern in this
paper.

However, $\bar{J}$ is the minimal value of $J$ only if  the matrices
$A, B$ and $Q$ in the ARE \eqref{eq:riccati} are determined ones. In
a formation system, the interactions between neighboring agents,
i.e., the offdiagonal entries in the system matrix $A$, are not any
intrinsic properties but rather some pending parameters to be
designed. Therefore, we focus primarily the following optimization
problem
\begin{align}
&J^*=\min_{A'}\sup\bar{J}(A)\label{eq:Jstar}\\
&\text{s.t. } u=-X^+x, \|x(0)\|_2=1\nonumber\\
&\text{\ \ \ \ \ }A=A_0+A', A_0=\diag(A)\nonumber
\end{align}
This paper aims to discuss this problem so as to exploit the
relationships between the interaction parameters $A'$ among agents
and the cooperative performance, and to further understand how to
design a formation system that meet the desired cooperative
performance in an optimistic way.

The following assumption is made true throughout the paper
\begin{assumption}\label{assum:local}
The local dynamics of each individual agent is fixed and known a
priori.
\end{assumption}

The system architecture we considered consists of two layers one of
which is the communication topology among agents as depicted by
graph $G_\alpha$. When agents are dynamically isolated, $G_\alpha$
is an empty graph. In the other layer, the controllers of individual
agents communicate with one another over an underlying graph denoted
by $G_\beta$.

\section{Optimal control of two agents formations}\label{sec:two}
In this section we focus on a simple system with two identical
agents to illustrate the novel optimization problem we consider.
More specifically, apart from minimizing  the control energy, we
also dive into the convergence process during destination attainment
and distinguish different kinds of cooperative performance that meet
various tasks requirements.

We consider two agents on positions $x_1$ and $x_2$ moving towards
their destinations $\bar{x}_1$ and $\bar{x}_2$ respectively along
parallel lines and in a cooperative manner, as shown in Fig.
\ref{fig:two_opt}. Generally there are two different cooperative
behaviors for the two agents one of which is the convergence to the
destination of each individual agent, as measured by the position
error $\delta x_i=x_i-\bar{x}_i, i=1,2$. The other one is the
formation maintenance between agents, as represented by the relative
formation error of each agent $i$: $|\delta x_1-\delta x_2|$. Three
kinds of performance could be considered concerning these two
behaviors:
\begin{itemize}
  \item Destination-first performance: each of the agents concentrates
  on moving towards its own destination in a selfish way. In Fig.
  \ref{fig:two_opt_goaled} agent 1 left agent 2 far behind although
  $x_1=\bar{x}_1$ at $t=1$.
  \item Formation-first performance: each of the agents pays too much
  attention to the other agent and tries to maintain the relative
  formation on its best efforts. In Fig. \ref{fig:two_opt_depend}
  agent 1 does not move
   until agent 2 catches up so that
  $|\delta x_1-\delta x_2|=0$.
  \item Neutral performance: each agents find itself a balance
  between the above two kinds of performance, i.e., efforts are put on formation
  maintenance and destination attainment simultaneously, as for
  example shown in Fig. \ref{fig:two_opt_illustrate}.
\end{itemize}
For the two agents formations, those three kinds of performance
could be captured by three parameters $\gamma_1,\gamma_2$ and
$\gamma_3$ in
\begin{equation}\label{eq:two_J}
  J_1=\int_0^\infty \gamma_1(\delta x_1-\delta x_2)^2+\gamma_2 \delta x_1^2+\gamma_3 \delta
  x_2^2 dt
\end{equation}
where $\gamma_1>0$ corresponds to the importance of formation
maintenance and $\gamma_2>0,\gamma_3>0$ correspond to the importance
of approaching the destination. Thus when $\gamma_1\ll \gamma_2$ and
$\gamma_1\ll\gamma_3$, the destination-first performance is
captured, and similarly the condition that $\gamma_1\gg \gamma_2$
and $\gamma_1\gg\gamma_3$ correspond to the formation-first
performance. When $\gamma_1\gamma_2\gamma_3\neq 0$, $J_1$ combines
both of the two kinds of performance with weights $\gamma_i$. In
this case, when $J_1$ reaches a stable value, $\delta x_1-\delta
x_2=0$, which indicates agent 1 and agent 2 has attained the
expected geometry, and $\delta x_1=\delta x_2=0$, which indicates
the achievement of the destination.

If we set the entries of the quadratic matrix $Q$ in \eqref{eq:J} as
$q_{11}=\gamma_1+\gamma_2$, $q_{12}=-\gamma_1$ and
$q_{22}=\gamma_1+\gamma_3$, $J_1$ then has the compact form of
$\int_0^{\infty}\delta x^TQ\delta x dt$.
 Thus the optimization problem concerning
the two-agent formation system could be stated as
\begin{figure}
\subfigure[Destination-first]{\label{fig:two_opt_goaled}
\begin{minipage}[b]{0.31\linewidth}
\centering
\includegraphics[scale=0.8]{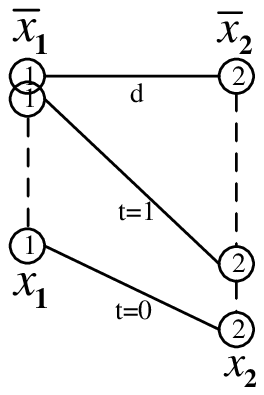}
\end{minipage}}
\subfigure[Formation-first]{\label{fig:two_opt_depend}
\begin{minipage}[b]{0.31\linewidth}
\centering
\includegraphics[scale=0.8]{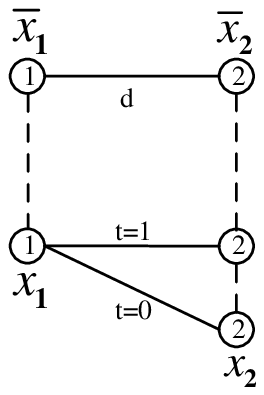}
\end{minipage}}
\subfigure[Neutral]{\label{fig:two_opt_illustrate}
\begin{minipage}[b]{0.31\linewidth}
\centering
\includegraphics[scale=0.8]{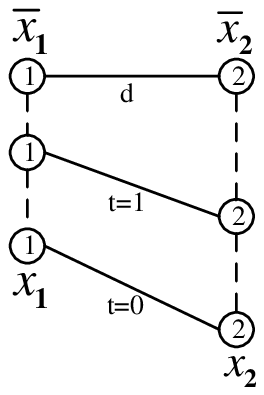}
\end{minipage}}
\caption{The three kinds of cooperative
performance}\label{fig:two_opt}
\end{figure}

\begin{problem}\label{problem:two}
Consider two agents with identical local dynamics
  \begin{equation}\label{eq:local_dyna}
    \dot{x}_i(t)=ax_i(t),i=1,2
  \end{equation}
 and are interconnected through parameters $r_1,r_2\in\mathbb{R}$. The overall system
 is modeled by
  \begin{equation}
  \bar{S}: \begin{bmatrix}\dot{x}_1\\ \dot{x}_2\end{bmatrix}=\begin{bmatrix}
    a&r_1\\r_2&a
  \end{bmatrix}\begin{bmatrix}
    x_1\\x_2
  \end{bmatrix}+\begin{bmatrix}
    u_1\\u_2
  \end{bmatrix}:=Ax+u
\end{equation}
Assume the quadratic matrix $Q=\begin{bmatrix} q&m\\m&q
  \end{bmatrix}>0$.
 Find the solution to the optimization problem
\begin{align}
      &\min\limits_{r_i\in\mathbb{R}} \|X^+\|_2&\nonumber\\
      &\text{s.t.\ \ }2a x_{11}-x_{11}^2+2r_2
      x_{12}-x_{12}^2+q=0&\label{eq:two1}\\
      &\text{\ \ \ \ \ \ }
      2r_1x_{12}-x_{12}^2+2ax_{22}-x_{22}^2+q=0&\label{eq:two2}\\
&\text{\ \ \ \ \ \ }
4ax_{12}-2x_{11}x_{12}+2r_2x_{22}-2x_{12}x_{22}+2r_1x_{11}+2m=0&\label{eq:two3}\\
&\text{\ \ \ \ \ \ }X^+=\begin{bmatrix} x_{11}&x_{12}\\x_{12}&x_{22}
\end{bmatrix} \in\mathbb{R}_S^+&\nonumber
    \end{align}
\end{problem}
\begin{remark}
When $Q$ is a positive definite matrix and when $B=I$, the existence
of the positive solution $X^+$ to the optimization problem is
guaranteed.
\end{remark}

The solution to the above optimization problem is discussed in
details in Section \ref{sec:multi}. Here we use a numerical example
to show that the optimal solution is obtained when $r_1\neq 0$ and
$r_2\neq 0$.
\begin{example}
Consider the cost matrix
$$Q=\begin{bmatrix}
  2 & -1\\-1& 2
\end{bmatrix}$$
According to \eqref{eq:two_J}, $Q$ indicates that the maintenance of
formation is as important as the convergence to the destination.

The local dynamics of each agent is uniquely chosen as a single
integral with $a=2$. First we consider agents with no interactions,
i.e., $A_1=2\diag(\univec)$ and $r_1=r_2=0$. By solving Problem
\ref{problem:two}, one can obtain the optimal upper bound of the
cost value $\bar{J}=4.65$.

On the other hand, if interconnections between agents are
considered, for example
$$A_2=\begin{bmatrix}
  2&1.2\\-1&2
\end{bmatrix},$$
the 2-norm of the solution $X^+$ is $\bar{\lambda}(X^+)=4.47$, which
is better than the previous one. The corresponding closed-loop
system matrix is
$$A_c=A_2-X^+=\begin{bmatrix}
  -2.46&1.21\\-0.98&-2.44
\end{bmatrix}$$
\end{example}

\begin{conjecture}If the local dynamics of agents are set to be $\delta
\dot{x}_i(t)=2\delta x_i(t),i=1,2$, for the worst case, the system
with interaction $A_2$ has better formation performance and consumes
less control energy than the situation when the two agents are
disjoint.
\end{conjecture}

\section{Optimal Formation Control of Multiple
Agents}\label{sec:multi} For a group of agents that is assigned with
cooperative payload transport tasks, each agent aims to approach its
destination while maintaining a certain geometry during the process.
In most cases, one may have to choose either the destination-first
performance which may result in formation failure during the
convergence, or the formation-first performance which may increase
the task execution time. Similar issues arise in other situations
such as formation reconfiguration and  obstacle avoidance. Thus a
balance between those two kinds of performance is demanding and is
worth of exploration.

Consider a set $\mathcal{C}=\{Q\in\mathbb{R}_S^+:
q_{ii}>0,q_{ij}<0\}$. Then the three kinds of performance introduced
in Section \ref{sec:two} could be represented by choosing
appropriate entices of $Q\in\mathcal{C}$.

Under Assumption \ref{assum:local}, we decouple the matrix into
$$A=A_0+A'$$
where $A_0=\bar{\Lambda}(A)$ if the dynamics of each agent is in two
dimensions and $A_0=\Lambda(A)$ if it is in one dimension.
Specially, when $A'=0$, all agents are dynamically isolated with
each other in the open loop.

Based on the LQR control, when focusing on the worst case, the
optimization problem is formalized as
\begin{problem}\label{pro:multi}
\begin{align}\label{eq:areB}
      & \min\limits_{A'} \|X\|_2&\nonumber\\
  &\text{s.t.\ \
      }A^TX+XA-XB^TBX+Q=0&\\
      &\text{\ \ \ \ \ \
      }A=A_0+A',~A_0=\bar{\Lambda}(A),~(A,B)~\text{controllable}&\nonumber\\
      &\text{\ \ \ \ \ \ }Q\in
      \mathcal{C},~X\in\mathbb{R}_S^+&\nonumber
    \end{align}
    \end{problem}
When $A'$ solves the optimization problem, $(A,Q)$ is called a
\emph{matched pair} and the corresponding positive solution of the
ARE \eqref{eq:areB} is denoted by $X^*(A,Q)$ which is sometimes
abbreviated by $X^*$.

Following are two well-proved lemmas from  \cite{Lancaster95ARE}
 that motivate our research.
\begin{lemma}\label{lemma:conti}
  Consider the algebraic Riccati equation \eqref{eq:areB}.
The unitary stable hermitian solution $X^+$ is a continuous function
of the 3-tuple $(A,Q,B^TB)\in \Phi$ where
$\Phi:\{(A,Q,B^TB)|Q\in\mathbb{R}^+_S, (A,B) \text{
stabilizable}\}$.
\end{lemma}
When $A,B,Q$ are continuous functions of a vector $\gamma$, so does
$X^+$.
\begin{lemma}
Let
$\Omega_\varepsilon(A,Q,B^TB)=\{(\bar{A},\bar{Q},\bar{B}^T\bar{B})|
\|\bar{A}-A\|+\|\bar{Q}-Q\|+\|\bar{B}^T\bar{B}-B^TB\|<\varepsilon\}$.
For every $(A,Q,B^TB)\in \Phi$ there exist positive constants $K$
and $\epsilon$ such that
\begin{align}
  &\|\bar{X}^+(\bar{A},\bar{Q},\bar{B}^T\bar{B})-X^+(A,Q,B^TB)\|\leq \nonumber\\
  &K(\|\bar{A}-A\|+\|\bar{Q}-Q\|+\|\bar{B}^T\bar{B}-B^TB\|)^{1/2}
\end{align}
is satisfied for all $(\bar{A},\bar{Q},\bar{B}^T\bar{B})\in
\Omega_\epsilon(A,Q,B)$. $X^+$ is said to be a locally Lipschitz
continuous function of order $1/2$.
\end{lemma}
These properties of the solution $X^+$ guarantee the feasibility of
the optimization problem \ref{pro:multi}.

Recall Problem \ref{pro:multi}, although the objective function is
convex, the constraints of the ARE \eqref{eq:areB} on arguments
$X^+$ and $A_0$ are nonlinear, and the variable $X^+$ is expected to
be symmetric, which falls into the category of the nonlinear
semidefinite programming(NSDP) problem. Meanwhile, this minimax
problem is nonsmooth on $X^+$, thus most of the analytical
optimization techniques are inapplicable. Although the nonlinear
optimization toolbox in Matlab may sometimes give structured
solutions on both $A'$ and $X^+$, generally special cares are
required so as to accelerate the optimization process and it may
sometimes fail.

On the other hand, even if there are some analytical iterative
methods to solve the optimization problem \ref{pro:multi}, the
optimal solution of $X^+$ and $A'$ are dense matrices which
indicates that the underlying graphs $G_\alpha$ and $G_\beta$ are
complete graphs. This is obviously not applicative especially for
large-scale formation systems. In the remaining context, we will
first consider a special case where agents are homogeneous ones.
Distributed controllers are considered that inherit the
communication graph $G_\alpha$ among agents: $G_\alpha=G_\beta$, as
for example shown in Fig. \ref{fig:struc_homo}. Further for the
general situation, we discuss the optimization problem with topology
constraints imposed in the agents layer. Examples of the underlying
graphs $G_\alpha$ and $G_\beta$ are illustrated in Fig.
\ref{fig:struc_heter}.
\begin{figure}
 \subfigure[Congruent underlying graph for system with homogeneous agents]{\label{fig:struc_homo}
\begin{minipage}[b]{0.45\linewidth}
\centering
\includegraphics[scale=0.55]{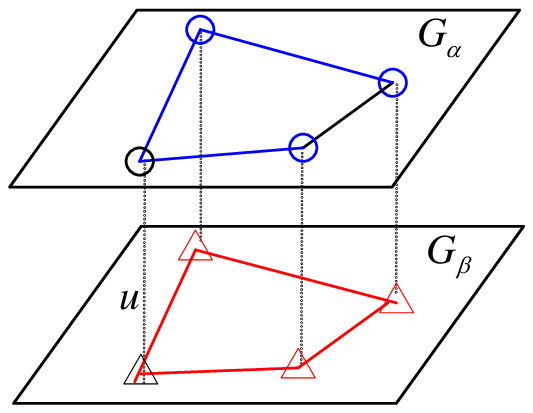}
\end{minipage}}
\subfigure[Different underlying graphs for system with heterogeneous
agents]{\label{fig:struc_heter}
\begin{minipage}[b]{0.45\linewidth}
\centering
\includegraphics[scale=0.55]{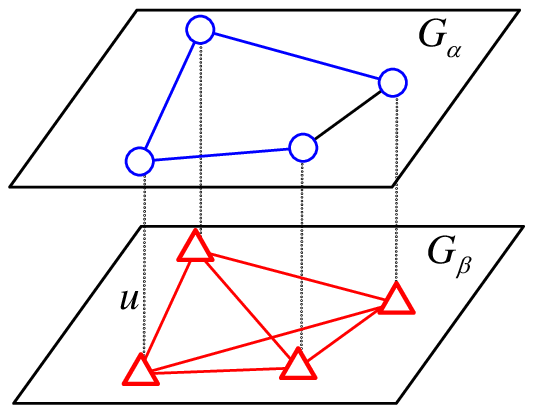}
\end{minipage}}
\caption{The underlying graphs of the two layers in a formation
system.}\label{fig:struc}
\end{figure}

\subsection{Distributed Controllers Design for Homogeneous Agents}
As mentioned in the preliminary, the formation system is broken down
into two layers each of which has its own communication topology. In
practical scenarios, agents and controllers always share
communication channels, therefore generally $G_\alpha=G_\beta$ is
more preferred. In this subsection, we consider formation systems
consist of homogeneous agents and let $G_\alpha=G_\beta\triangleq
G_h$. The following assumptions are also made
\begin{assumption}\label{assumption:homo}
In a formation system
\begin{itemize}
\item the dynamics of each agents is decoupled along $x$ and $y$
axis such that we can focus on the one-dimensional model of each
agent
\item the local dynamics of agents are congruent: $$\dot{x}_i=ax_i$$
\item the interaction parameters among any pairs of neighboring agents are
congruent:
\begin{equation}\label{sys:identical}
\dot{x}_i=ax_i+b\sum_{j\sim i}x_j
\end{equation}
\item the input matrix $B$ is an identical matrix. This assumption is realistic and is easy
to be satisfied in a formation system
\item the quadratic matrix $Q$ has identical diagonal entries and
identical nondiagonal entries, i.e.,
$$\diag(Q)=q\univec,~Q-\Lambda(Q)=p\mathcal{A}$$
where $\mathcal{A}$ is the adjacency matrix of the assigned graph
$G_h$.
\end{itemize}
\end{assumption}
Assumption \ref{assumption:homo} allows us to decouple the overall
formation system into $n$ independent subsystems, and thus the
design of distributed controllers is also broken down accordingly,
which leads to the following controller design method.
\begin{theorem}\label{theorem:homo}
  Consider system \eqref{sys:identical} and the Assumption
  \ref{assumption:homo}. The optimal solution $X^*$ to Problem \ref{pro:multi} is
  \begin{equation}\label{eq:homo_controller}
    X^*=xI+y\mathcal{A}
  \end{equation}
where $$x=a+\sqrt{a^2+q},~~y=\frac{2p}{2x-4a}$$ and is the optimal
controller when when $b=\frac{p}{2x-4a}$.
\end{theorem}
\begin{proof}
  Under Assumption \ref{assumption:homo}, the system matrix $A$ and
  quadratic matrix $Q$ has the special form of
  \begin{equation}\label{eq:homoAQ}A=aI+b\mathcal{A},Q=qI+p\mathcal{A}\end{equation}
Assume the optimal solution has the form $X^*=xI+y\mathcal{A}$ such
that the ARE is rewritten as
  $$(aI+b\mathcal{A})(xI+y\mathcal{A})+(xI+y\mathcal{A})(aI+b\mathcal{A})-(xI+y\mathcal{A})^2+(qI+p\mathcal{A})=0$$
which has the equivalent expression of
\begin{equation}\label{eq:temp}
2(axI+bx\mathcal{A}+ay\mathcal{A}+by\mathcal{A}^2)-(x^2I+2xy\mathcal{A}+y^2\mathcal{A}^2)+(qI+p\mathcal{A})=0
\end{equation}
Let $P$ be the
  orthogonal permutation matrix of $\mathcal{A}$ such that
  $$\Lambda_\mathcal{A}=P^T\mathcal{A}P$$
  where the diagonal entries are the eigenvalues of the adjacency
  matrix $\mathcal{A}(G)$. Left and right multiplying the left hand side expression in equation \eqref{eq:temp} by $P^T$ and $P$
yields
\begin{equation}\label{eq:homo_inthm}
(2ax-x^2+q)I+(2bx+2ay-2xy+p)\Lambda_\mathcal{A}+(2by-y^2)\Lambda_\mathcal{A}^2=0
\end{equation}
Due to the feasibility of the ARE and the positive semidefinite
constraints on the solution $X^*$, the parameters $x,y,b$ are
\begin{align}\label{eq:homox}
&x=\frac{-2a- \sqrt{4a^2+4q}}{-2}\nonumber\\
&b=\frac{p}{2x-4a}\nonumber\\
&y=2b
\end{align}
Substituting $b$ by $y/2$ back into equation \eqref{eq:homo_inthm}
and with some trivial calculations, we obtain that
\begin{equation}\label{eq:homo_temp1}(2a-x)(x+\lambda_i(\mathcal{A})y)+q+\lambda_i(\mathcal{A})p=0\end{equation}
Due to the expression of $x$ in \eqref{eq:homox}, $2a-x<0$. As the
quadratic matrix $Q>0$, it is also true that
$q+\lambda_i(\mathcal{A})p>0$ for all $i$. Thus equation
\eqref{eq:homo_temp1} is satisfied if and only if
$$xI+y\mathcal{A}\in\mathbb{R}^+_s$$
This infers that $X^*$ with $x$ and $y$ given in \eqref{eq:homox} is
the positive definite solution and is the optimal controller when
$b=\frac{p}{2x-4a}$.
\end{proof}
\begin{remark}
   Controller $K=-X^*$ is the stabilization controller with underlying
graph generated from $\mathcal{A}$. Thus the communication
structures among agents and controllers are congruent ones, i.e.,
$G_\alpha=G_\beta$ as, for example, the one in Fig.
\ref{fig:struc_homo}.
\end{remark}
\begin{remark}
  The synthesis of controller for each
  agent is carried out in a distributed way and the control signal
  passed to each agent is generated using information from the neighbors. This property provides the
  formation system the feature of scalability.
\end{remark}

In \cite{Borrelli08LQR} the authors focused themselves on the
scalable synthesis of distributed controllers for large-scale
identical isolated subsystems, i.e. $A'\equiv 0$, and further
explored the relationship between the local controllers and the
stability of the overall system. Next we will show that, compared to
the formation system where $A'\equiv 0$, a set of carefully designed
interaction parameters would indeed give the overall system a better
cooperative performance and require less control energy.

\begin{corollary}\label{theorem:A'=0Q'=0}
Consider a formation system consists of $N$ agents under Assumption
\ref{assumption:homo}. For a matched pair $(A_0+A',Q)$, $A'=0$ if
and only if $Q\in \mathbb{R}_\Lambda$.
\end{corollary}
\begin{proof}
Recall \eqref{eq:homo_inthm}, if $p=0$, due to the positiveness of
$q$, it is necessarily that $y=0$ and $b=0$, which indicate $A'=0$.
According to the solutions $x,y$ and $b$ in Theorem
\ref{theorem:homo}, when $b=0$, $q=0$, which proves the sufficient
condition.
\end{proof}

When $q_{ij}\equiv 0, i\neq j$, it means no cooperative behavior is
expected in a formation system. Corollary \ref{theorem:A'=0Q'=0}
indicates that in such cases system with disjoint agents can perform
better. However, due to the inverse negative proposition of
Corollary \ref{theorem:A'=0Q'=0}, as long as there exists
$q_{ij}\neq 0$, having agents communicate with others at some
appropriate parameters in the open loop is a better choice. Note
that these results only apply to system with identical agents and
with equivalent weights being put on the diagonal of the cost matrix
$Q$. Systems with heterogeneous agents are not eligible to this
property. Actually, when $a_{ii}$ is not identical, interactions are
required even if $q_{ij}=0$  so as to make up a matched pair.

For a set of homogeneous agents, by finding the appropriate
parameter $b$ between each pair of neighboring agents, we developed
the optimal state-feedback controllers that inherits exactly the
same underlying structure as the agents, and the upper bounds of the
quadratic cost function of the formation system is minimized.

\subsection{Heterogeneous Formation Systems under Structure Constraints}
In this subsection, we further consider a more general case where
all those restrictions in Assumption \ref{assumption:homo} no longer
hold true.  When the dynamics of agents is coupled along $x$ axis
and $y$ axis, the dimension of the overall system model is $2n$
instead of $n$ and thus the local dynamics as well as the
interaction parameters between neighbors are captured by $2\times 2$
matrices. In order to find optimal interaction parameters for a
formation system where agents communicate over an assigned graph
$G$, we define a set of structured matrices whose underlying graph
is $G$:
$$T_G=\{M\in \mathbb{R}^{2n\times 2n}: M_{\bar{ij}}=M_{\bar{ji}}=\mathcal{A}_{ij}M_{\bar{ij}},i\neq j\}$$
where $\mathcal{A}$ is the adjacency matrix of $G$.

 The state-space
model of the overall system is
\begin{equation}
  \dot{p}=Ap+Bu
\end{equation}
where $A\in T_G\cap \mathbb{R}_S$ and $B\in T_G\cap \mathbb{M}$.
Note that for the general case, we restrict ourselves to systems
where interactions among agents are independent of the transmission
direction, i.e., matrix $A$ is symmetric. This is only for the
purpose of simplicity and all the results bellow apply to
nonsymmetric case as well. Meanwhile, the structure of the input
matrix indicates that agents exchanges their states and their
control inputs simultaneously with the neighbors.

As mentioned before, the most difficult part when dealing with
Problem \ref{pro:multi} is the nonlinear constraints in the ARE.
This inspires us to exploit methodologies to eliminate this
restriction.

What we focus on in this paper is the two norm of $X^+$. We noticed
that the upper bound of the positive solution to the ARE
\eqref{eq:areB} has been discussed widely by researchers in the
related fields. Quite a few upper bounds are given in literatures
among which the following result borrowed from
\cite{Kang96AREupperbound} is known as a relatively tighter bound
than many others.
\begin{lemma}\label{lemma:XupperBound}
Consider ARE \eqref{eq:areB}. The norm of the positive semidefinite
solution $X^+$ is upper bounded by
  \begin{equation}\label{eq:XupperBound}
\|X^+\|\leq \|P\|^{1/2}\|APA+Q\|^{1/2}+\mu(AP)
\end{equation}
where $\mu(AP)=1/2\bar{\lambda}(AP+P^TA)$ and $P=(B^TB)^{-1}$.
\end{lemma}
\begin{remark}
This  upper bound of $\|X^+\|$ also validates the fact that the
minimal value of $J$ is correlated with $A$ and $Q$, and when $Q$
and $B$ are not diagonal matrices, the minimal upper bound of
$\|X^+\|$ is the one when $A$ is nondiagonal as well.
\end{remark}

We then propose the following relaxed optimization problem that is
parallel to Problem \ref{pro:multi}, but the nonlinear constraint is
avoided. Note that as we consider formation on a plane, the local
dynamics between each pair of agents is represented by  a square
matrix $A'_{\bar{ij}}$.
\begin{problem}\label{pro:multi_weight}
\begin{align}
  &\min_{A'} \phi(A')=\|P\|^{1/2}\|APA+Q\|^{1/2}+\frac{1}{2}\bar{\lambda}(AP+PA)\label{opt:weight}\\
  &\text{s.t.~~} A=A_0+A',~A_0=\bar{\Lambda}(A),~A\in T_G\cap \mathbb{R}_S,\nonumber\\
  &~~~~~~Q\in
  T_G\cap \mathcal{C},~B\in T_G\cap \mathbb{M},~P=(B^TB)^{-1}\nonumber
\end{align}
\end{problem}
The minimum of $\phi(A')$ is calculated when the the argument $A'$
fits into the desired underlying graph $G$. This minimum could be
further reduced when this structure restriction is abandoned, which
is, however, accompanied by the increase on the communication cost.
\begin{theorem}
  The optimization problem \eqref{opt:weight} is convex.
\end{theorem}
\begin{proof}
  It is well-known that $\bar{\lambda}(A)$ is a convex function on
  $A$. Due to the positiveness and symmetry of $APA+Q$, $\|APA+Q\|=\bar{\lambda}(APA+Q)$.
  The conclusion is then self-evident.
\end{proof}
Despite of the convexity, the function $\phi(A')$ is non-smooth and
thus is non-differential, which requires some special optimization
techniques. In our case, the constraints are linear matrices
constraints, thus it is then an ordinary linear programming(LP)
problem. Interior-point methods and subgradient methods are two of
the frequently used techniques when dealing with LP, as in
\cite{Jakovetic10weightOptimization,Kar08topologyDesign,Boyd06randomizedGossipAlgorithm,XiaoBoyd04fastLinearIterations}.
More specifically, a similar problem is discussed in
\cite{Jakovetic10weightOptimization} where the objective function
$\phi$ is the expectation of the system matrix. Due to this
similarity, here we also adopt the subgradient algorithm  as in
\cite{Jakovetic10weightOptimization} to deal with this structure
constrained optimization problem.

Matrix $A$ is actually an affine function of the interaction
parameters $A_{ij},i\sim j$. Assume the underlying graph $G$ has $m$
edges indexed from $1$ to $m$. Define a vector $x\in\mathbb{R}^{4m}$
that is partitioned into $m$ sub-vectors
$x_k=\begin{bmatrix}x_k^1&x_k^2&x_k^3&x_k^4\end{bmatrix}^T,
k\in[1,m]$. The entries $x_k^p,p\in[1,4]$ indicate the interactions
between agent $i$ and agent $j$ that are connected by edge $l$:
$$\begin{bmatrix}
  x_k^1 & x_k^2\\
x_k^3 & x_k^4
\end{bmatrix}=A_{\bar{ij}}$$
Then the affine form of $A$ on vector $x$ is
\begin{equation}\label{WET:A_affine}
  A(x)=A_0+\sum^m_{\substack{
k=1,j> i\\j\sim i}} (x_k^1 E_{\bar{ij}}^1+
  x_k^2 E_{\bar{ij}}^2+ x_k^3 E_{\bar{ij}}^3 +x_k^4
  E_{\bar{ij}}^4)
  \end{equation}
  where
  \begin{align*}
  &E_{\bar{ij}}^1=e_{2i-1}e_{2j-1}^T+e_{2j-1}e_{2i-1}^T\\
&E_{\bar{ij}}^2=e_{2i-1}e_{2j}^T+e_{2j}e_{2i-1}^T\\
&E_{\bar{ij}}^3=e_{2i}e_{2j-1}^T+e_{2j-1}e_{2i}^T\\
&E_{\bar{ij}}^4=e_{2i}e_{2j}^T+e_{2j}e_{2i}^T
\end{align*}
The partial differential of $A'(x)$ or $A(x)$ with respect to
$x_k^p$ is
\begin{equation}
  \frac{\partial A(x)}{\partial x_k^p}=E_{\bar{ij}}^p, p\in\{1,2,3,4\}
\end{equation}
and further
\begin{equation}
  L(x_k^p):\triangleq \frac{\partial (APA+Q)}{\partial x_k^p}=APE_{\bar{ij}}^p+E_{\bar{ij}}^pPA, p\in\{1,2,3,4\}
\end{equation}
The Hessian of $\phi(A')$ is then
\begin{align}
  H_{ij}^p|_{x=x_k}&=Mv^T\frac{\partial(APA+Q)}{\partial
  x_k^p}v+\frac{1}{2}u^T\frac{\partial (AP+PA)}{\partial x_k^p}u\nonumber\\
  &=Mv^TL(x_k^p)v+\frac{1}{2}u^T(E_{\bar{ij}}^pP+PE_{\bar{ij}}^p)u
\end{align}
where $M=\frac{\sqrt{\|P\|}}{2\sqrt{\|APA+Q\|}}$. Vector $v$ is the
unit eigenvectors associated with the largest eigenvalue of $APA+Q$
and $u$ is the one associated with $\bar{\lambda}(AP+PA)$.
 The optimization algorithm is then carried out by computing the
subgradient of $\phi(A^k)$ at $A^k$ on each step with stepsize
$\gamma^k>0$. The stepsize satisfies the diminishing rule:
$$\lim_{k\to\infty} \gamma^k=0,\sum_1^\infty \gamma^k=\infty$$
As pointed out in \cite{XiaoBoyd04fastLinearIterations}, for the
convex optimization problem, the convergence of the algorithm is
well-known.
\begin{algorithm}
   {\textbf{Initialize} $A'^1\in T_G$, $B\in T_G$ and $Q\in
  T_G\cap \mathcal{C}, ~A$, $i=1$, $\epsilon>0$\;
\textbf{Repeat}\;
    \ \ \ Compute $M$\;
    \ \ \ Compute $v^i$ and $u^i$\;
    \ \ \ Compute $H^i$ at $x^i$\;
    \ \ \  Set $A'^{i+1}=A'^i-\gamma^i H^i$\;
    \ \ \  $i:=i+1$\;}
\textbf{Until} $|\Delta\phi(A')|<\epsilon$\;
  \caption{Optimal Parameters Design}\label{alg1}
\end{algorithm}

The optimal parameters generated from Algorithm \ref{alg1} minimize
the upper bound of $\|X^*\|$ in the ARE \eqref{eq:areB}. We point
out that there indeed have certain cases where a lower upper bound
may relate to a higher value of $\|X^*\|$. However the original
Problem \ref{pro:multi} is a great trouble while this \emph{relaxed}
optimization problem could be carried out using the reliable
subgradient method. Meanwhile, in most of cases we experimented,
lower upper bound of $\|X^*\|$ always corresponds to a smaller value
of $\|X^*\|$.

Substituting the structured matrix $A$ generated from Algorithm
\ref{alg1} into the ARE \eqref{eq:areB}, the resulting LQR
controller $K=-B^TX^*$ is possibly a dense matrix without any
structure constraints because of the density of $X^*$. In general,
$X^*$ is usually a matrix with nonzero entries even when $A$ is a
sparse matrix, as for example the ones in Fig.
\ref{fig:struc_heter}. This result is consistent to
\cite{Langbort09minimal} where when no communication expenditure was
considered, the undirected underlying graph of the optimal
interaction graph among the LQR controllers was typically a complete
graph. This reveals that a system with distributed controllers
consumes less energy compared with a decentralized one. However a
complex topology is obviously not practical, especially when the
scale of the system grows. In the above research, due to the
uniqueness of the positive solution of the ARE \eqref{eq:areB}, it
is impossible to find any structured solution so as to reduce the
communication burden among controllers. Other methodologies are
required to deal with this deficiency, which is our ongoing works.

In this section, the optimization problem is built up with the
objective function being nonsmooth and with nonlinear constraints.
Structure restrictions on the underlying graphs are introduced for
the purpose of realtime applications. To find an analytical solution
to the optimization problem, a special formation system with
homogeneous agents is considered. The system is decoupled into $n$
independent subsystems and the desired underlying graphs are imposed
both among agents and among controllers. For the more general case,
a relaxed optimization problem is proposed that eliminates the
nonlinear constraints and thus is solved based on the subgradient
method. The underlying graph $G_\alpha$ is designed to be the
desired one.

\section{Examples}
The proposed algorithms are validated on systems with identical
agents and systems with heterogeneous agents respectively. In order
to observe the cooperative performance of the relative formations
during convergence more intuitively, we consider the relative
formation system in the edge space that is transformed through
$$e=(\bar{H}\otimes I_2)p$$
where $\bar{H}$ is the oriented incidence matrix of a rigid graph
$\bar{G}=(\bar{V},\bar{E})$ and $p$ is the coordinates of agents.
For the desired formation $p_d\in\mathbb{R}^{2n}$ of $n$ agents and
the corresponding $e_d\in\mathbb{R}^{2|\bar{E}|}$, we define a
relative formation error function with respect to $e_d$ as
\begin{equation}\label{eq:fe}
fe(t)=\sum\left|\frac{\|e_i\|}{\|e_1\|}-\frac{\|e_{d_i}\|}{\|e_{d_1}\|}\right|
\end{equation}
When $fe=0$, the geometry formed by the agents coincide with the
desired geometry although their actual coordinates may differ. A
larger relative formation error corresponds to greater deformation
effects of a geometry with respect to the desired one.
\begin{figure}
 \subfigure[$G_a$]{\label{fig:topologya}
\begin{minipage}[b]{0.45\linewidth}
\centering
\includegraphics[scale=0.55]{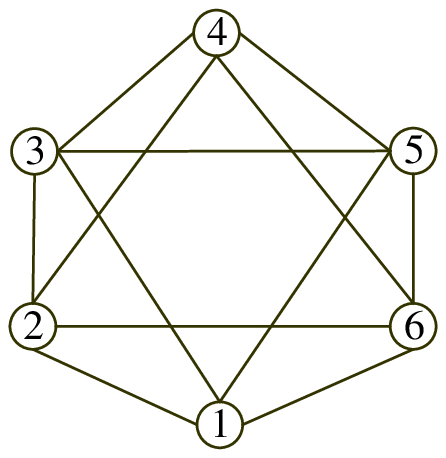}
\end{minipage}}
\subfigure[$G_b$]{\label{fig:topologyb}
\begin{minipage}[b]{0.45\linewidth}
\centering
\includegraphics[scale=0.55]{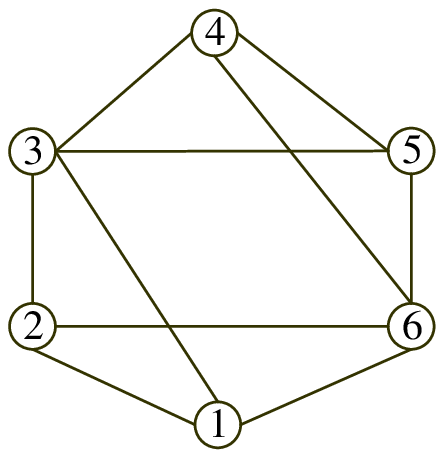}
\end{minipage}}
\caption{Two graphs with six vertexes}\label{fig:topology}
\end{figure}

\subsection{Formations with homogeneous agents}
For a formation system with identical agents and identical
interaction parameters, Theorem \ref{theorem:homo} provides an
analytical method to design distributed controllers so as to
minimize the cost function under the worst case.

The underlying graph is the one in Fig. \ref{fig:topologya}. The
parameters in \eqref{eq:homoAQ} are $a=2,q=8$ and $p=-2$. According
to Theorem \ref{theorem:homo}, the distributed controller $X^*$ has
identical diagonal entries of 5.46 and nondiagonal entries of -1.37.
The trajectories of the six agents are shown in Fig.
\ref{fig:trjc_six_homo} with snapshots in Fig.
\ref{fig:homo_snapshot}. Fig. \ref{fig:pe_homo} demonstrates the
position errors of the six agents with respect to the desired
positions. It takes approximately 8 seconds for them to achieve the
desired positions. However, according to Fig. \ref{fig:fe_homo}, at
approximately $t\approx2sec$, the relative formation error converges
to zero, which indicates that with high attentions being paid to the
relative formations of the six agents, they attain the desired
geometry before arriving at the destination.

In order to observe different kinds of cooperative performance, the
six agents are initialized at the same spot as shown in Fig.
\ref{fig:homo_stratage}. When $q=8,p=-2$, relative formation is more
important during the assignment, thus agents spread immediately
after they set off and maintain the desired geometry when
approaching the destination. On the other hand,  if the
destination-first performance is more impotent in the mission, we
set $q=3$ and $p=-0.2$. Then the six agents approach their own
destinations radially and form the desired geometry at almost the
same time as they achieve the destinations at $t\approx 3sec$. The
formation errors for the two kinds of performance are shown in Fig.
\ref{fig:homo_compare} in blue solid lines and red dashed lines
respectively. This infers that the parameters in the cost matrix $Q$
allow us to take into account the geometrical performance during
convergence.

\begin{figure}
\begin{center}
\includegraphics[scale=0.4]{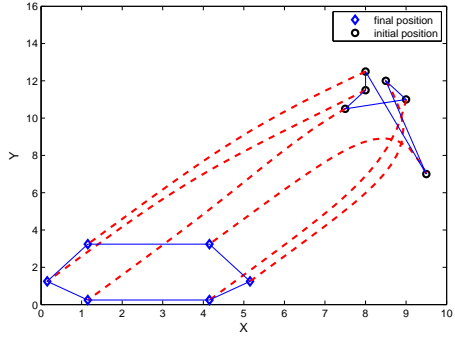}
\caption{Formations of six agents}\label{fig:trjc_six_homo}
\end{center}
\end{figure}

\begin{figure}
 \subfigure[$t=0.68s$]{
\begin{minipage}[b]{0.23\linewidth}
\centering
\includegraphics[scale=0.2]{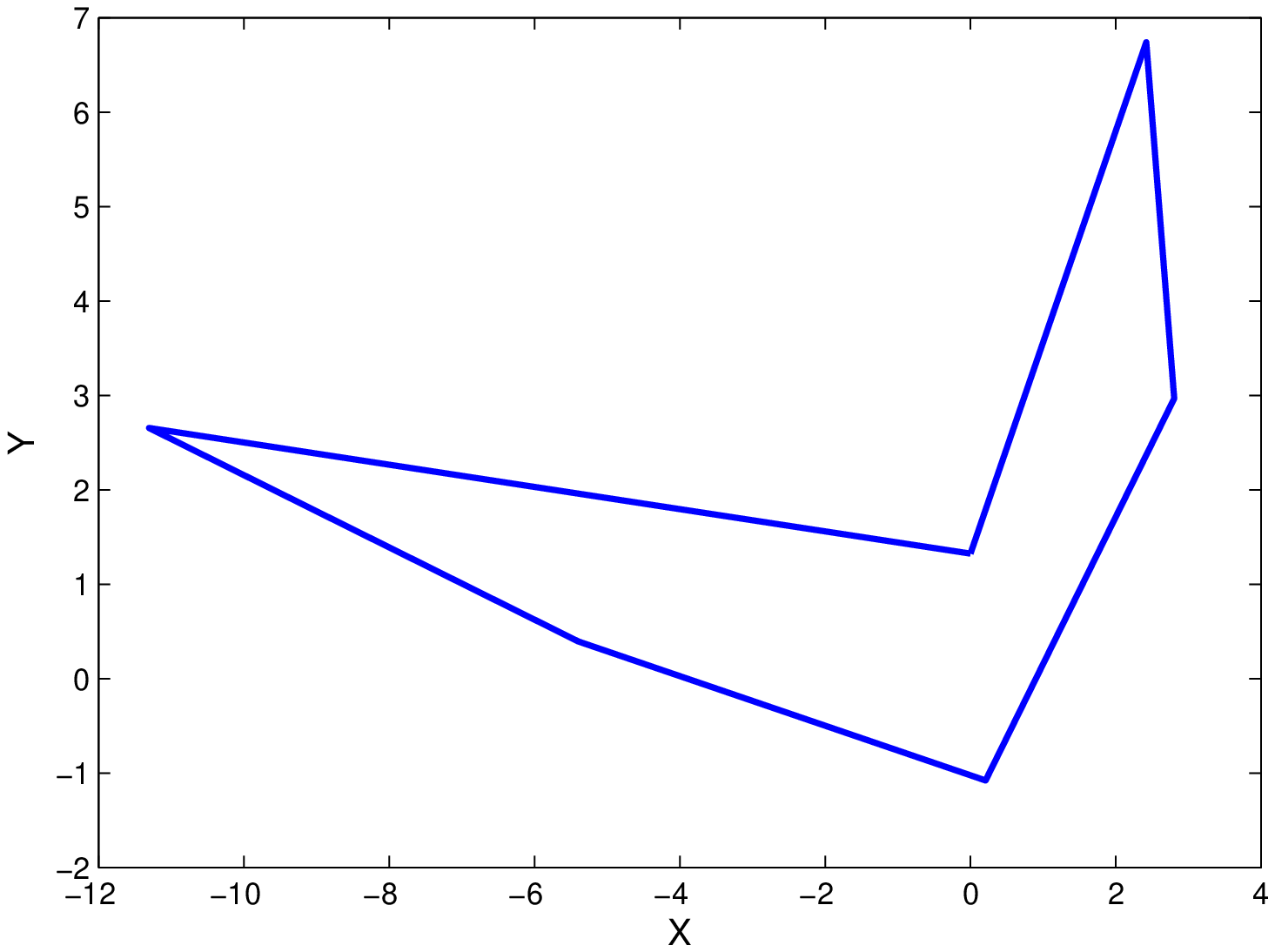}
\end{minipage}}
\subfigure[$t=1.48s$]{
\begin{minipage}[b]{0.23\linewidth}
\centering
\includegraphics[scale=0.2]{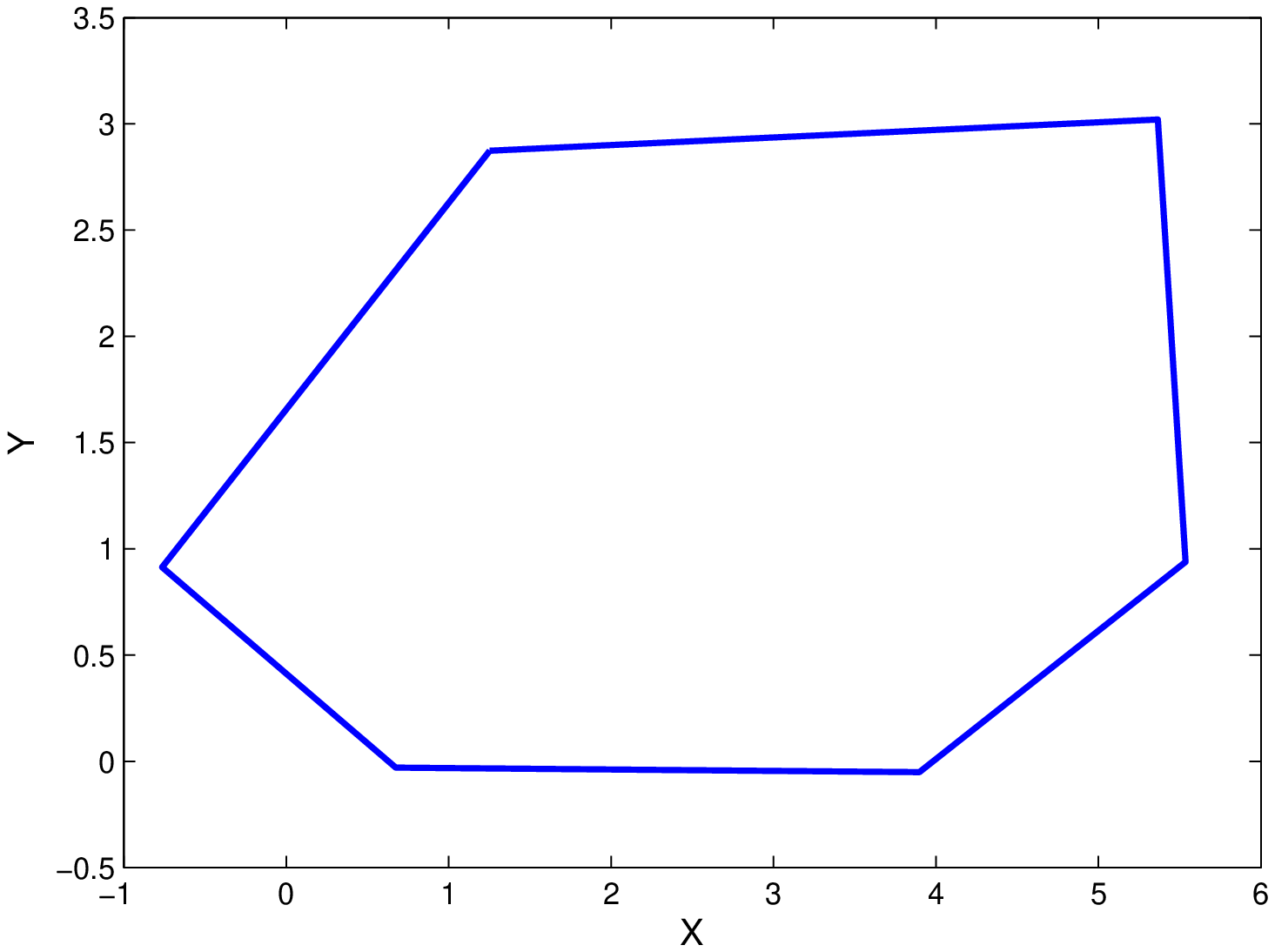}
\end{minipage}}
\subfigure[$t=2.28s$]{
\begin{minipage}[b]{0.23\linewidth}
\centering
\includegraphics[scale=0.2]{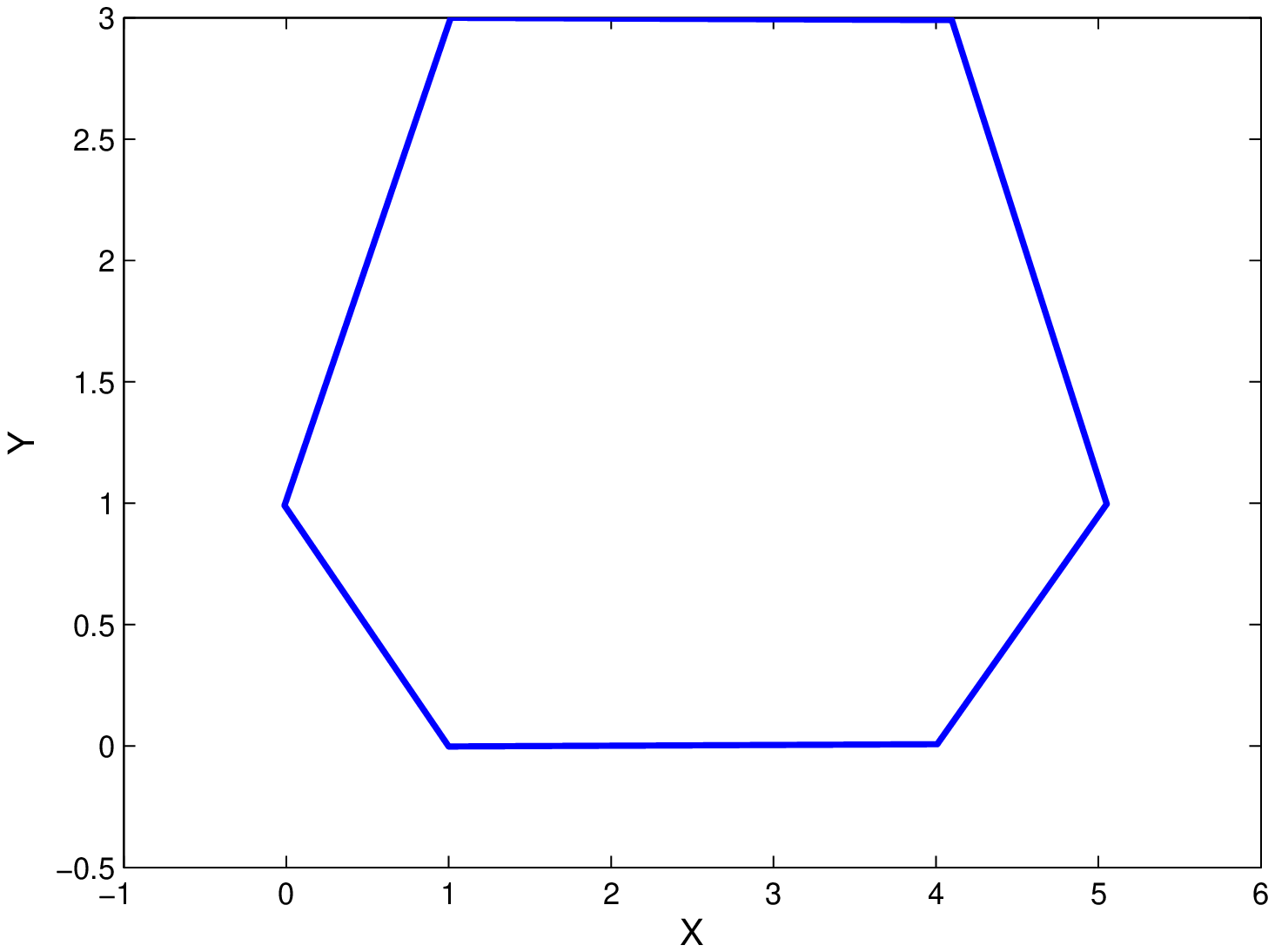}
\end{minipage}}
\subfigure[$t=3.08s$]{
\begin{minipage}[b]{0.23\linewidth}
\centering
\includegraphics[scale=0.2]{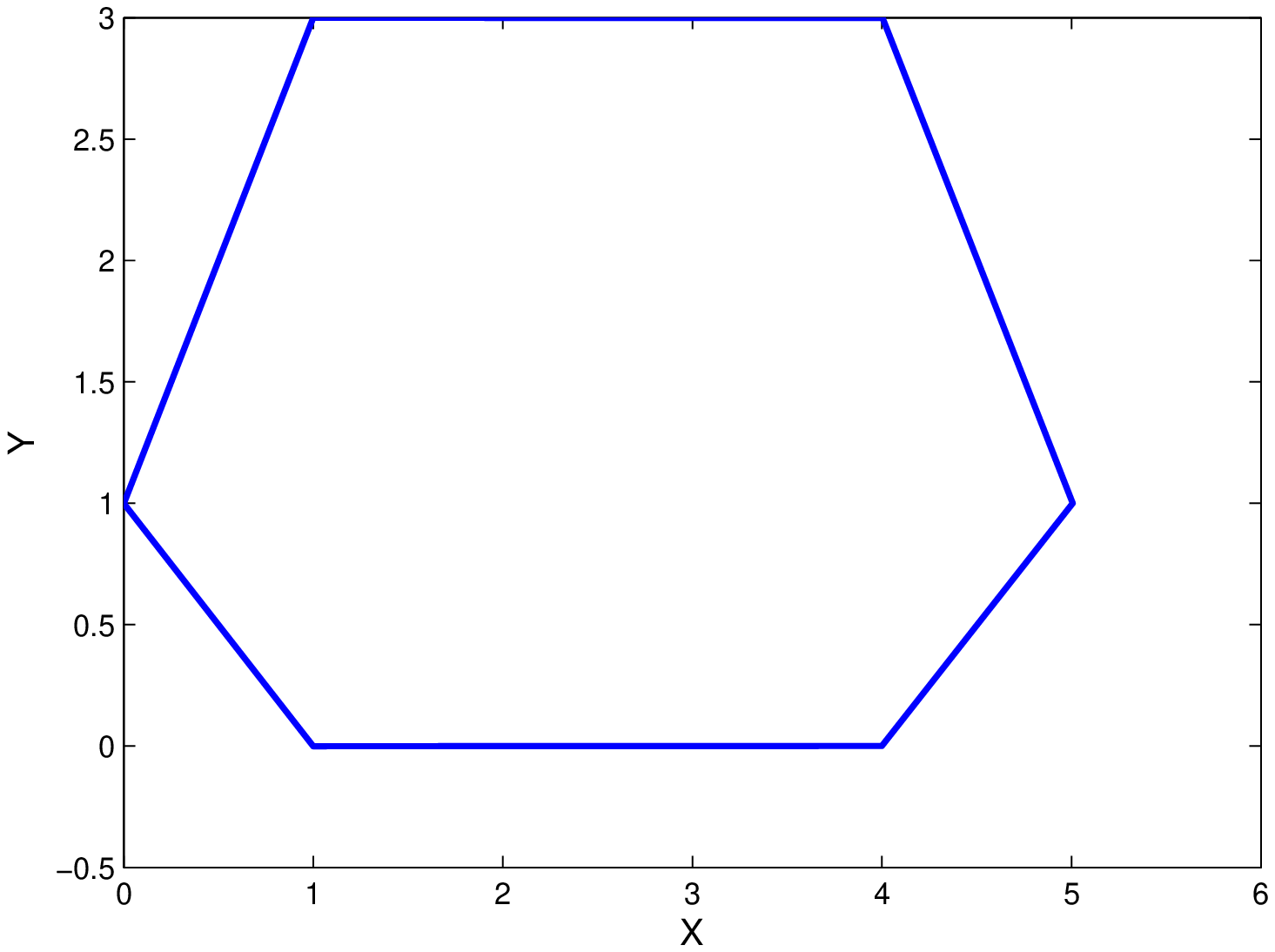}
\end{minipage}}
\caption{Snapshots of the six agents}\label{fig:homo_snapshot}
\end{figure}

\begin{figure}
 \subfigure[Position errors]{\label{fig:pe_homo}
\begin{minipage}[b]{0.45\linewidth}
\centering
\includegraphics[scale=0.25]{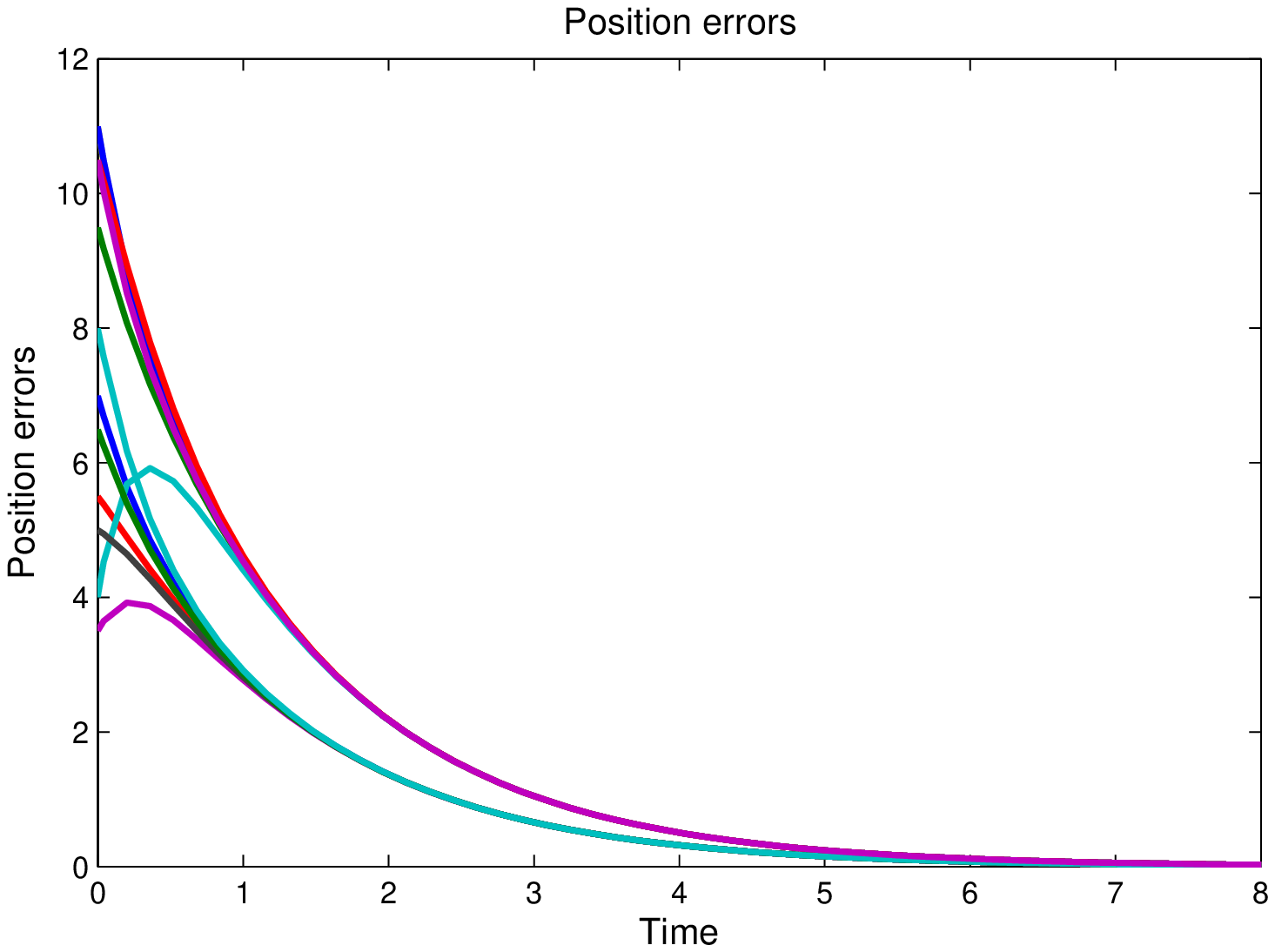}
\end{minipage}}
\subfigure[Formation errors]{\label{fig:fe_homo}
\begin{minipage}[b]{0.45\linewidth}
\centering
\includegraphics[scale=0.25]{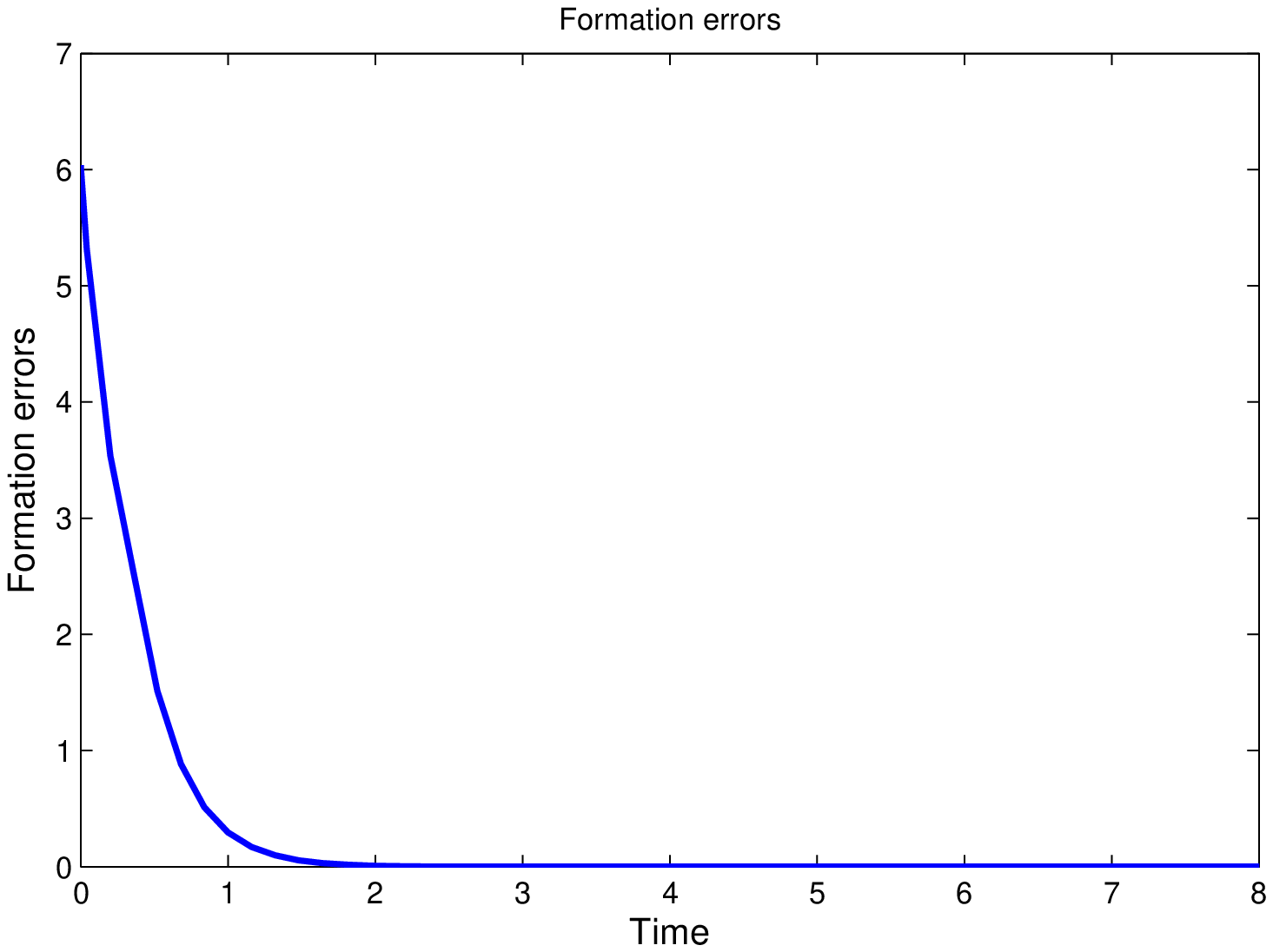}
\end{minipage}}
\caption{Position errors and  formation errors of the six
agents}\label{fig:homo_errors}
\end{figure}

\begin{figure}
 \subfigure[$q=8,p=-2$]{\label{fig:homo_onespot}
\begin{minipage}[b]{0.45\linewidth}
\centering
\includegraphics[scale=0.25]{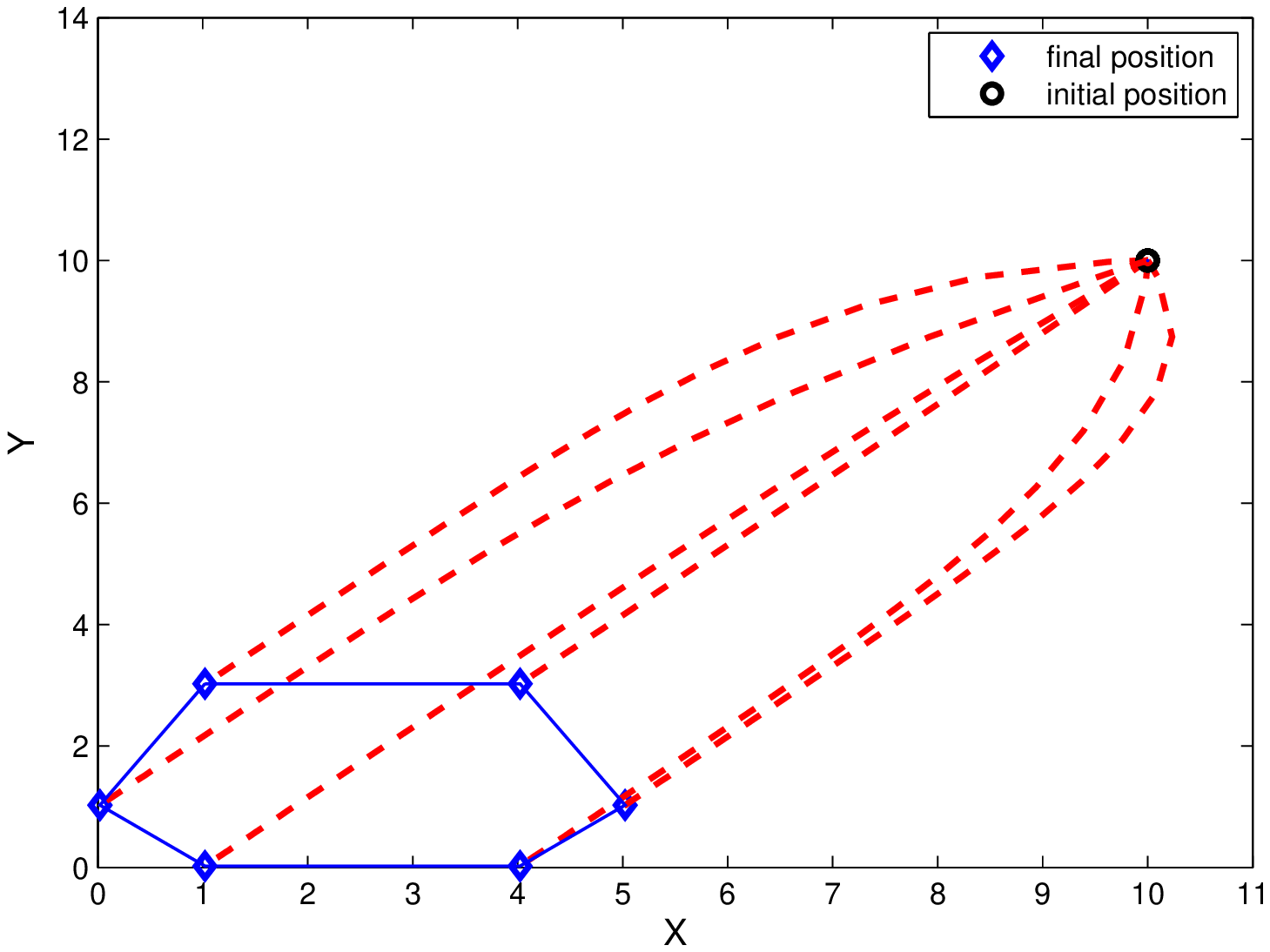}
\end{minipage}}
\subfigure[$q=3,p=-0.2$]{\label{fig:homo_onespot_self}
\begin{minipage}[b]{0.45\linewidth}
\centering
\includegraphics[scale=0.25]{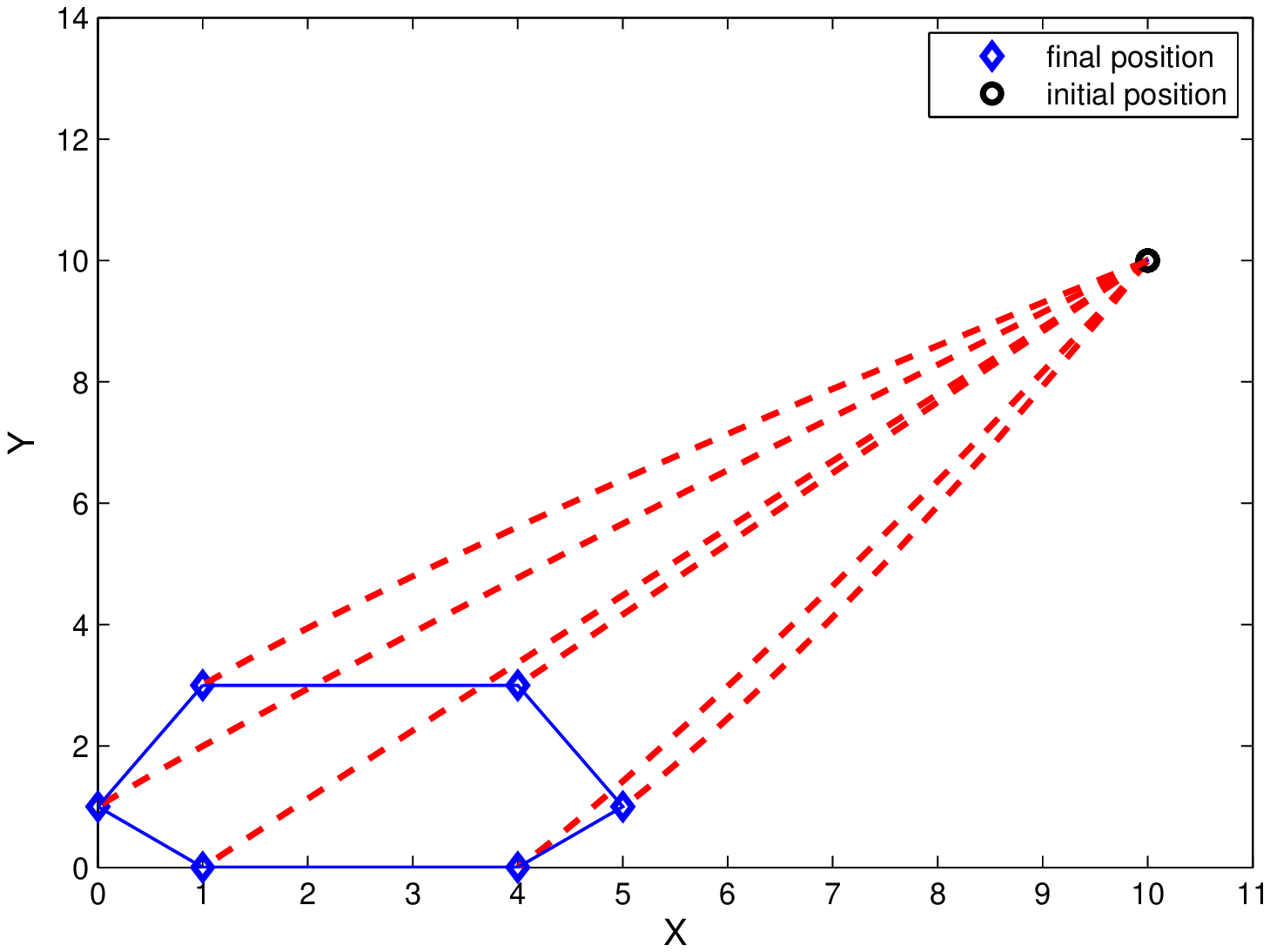}
\end{minipage}}
\caption{Formation-first performance  and the destination-first
performance}\label{fig:homo_stratage}
\end{figure}

\begin{figure}
\begin{center}
\includegraphics[scale=0.4]{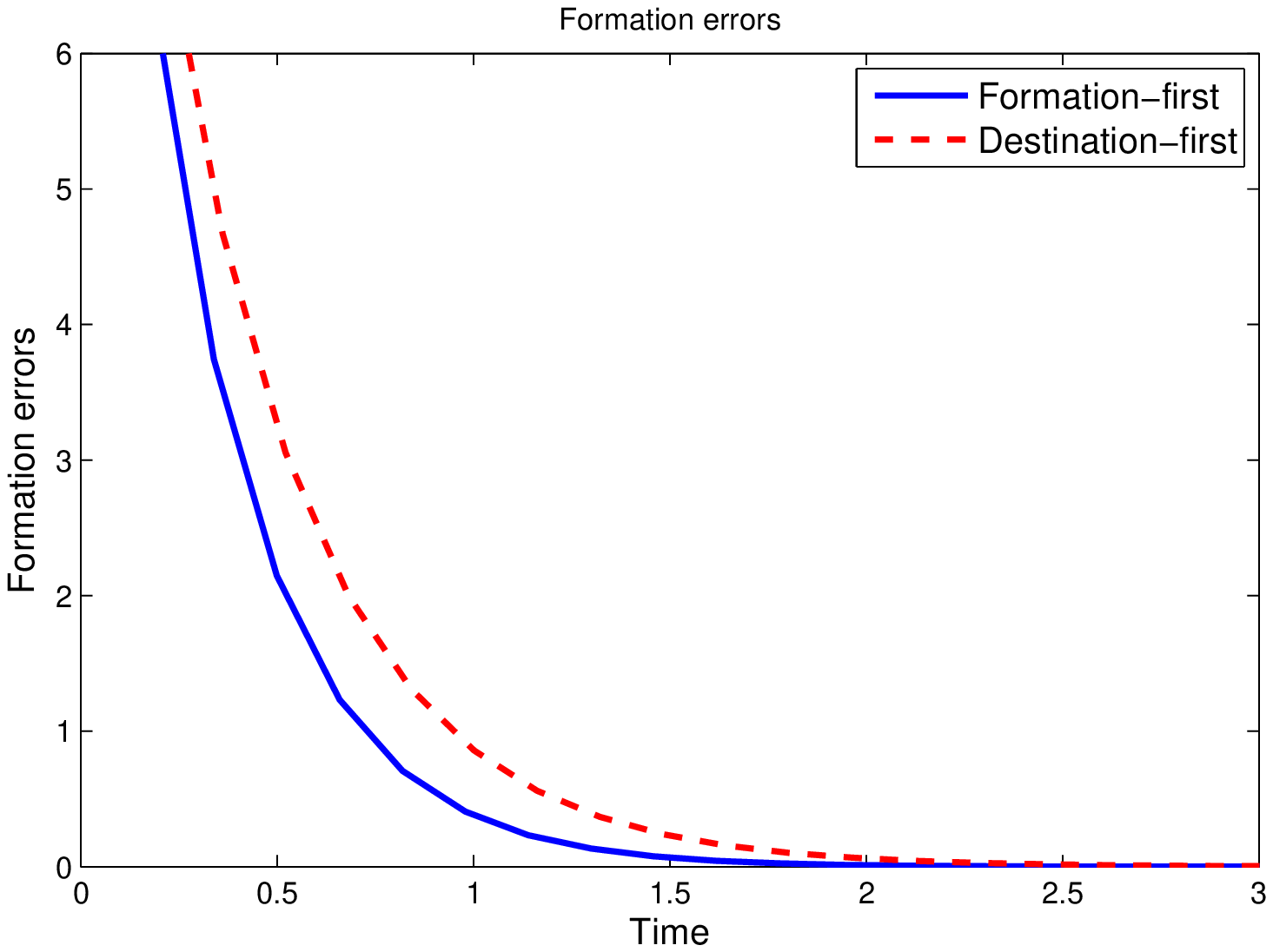}
\caption{Formation errors for different kinds of cooperative
performance}\label{fig:homo_compare}
\end{center}
\end{figure}

\subsection{Formations with heterogeneous agents}
When agents exhibit different local dynamics and when the
interactions between neighboring agents various from one another,
Algorithm \ref{alg1} is proposed to solve for the set of optimal
interaction parameters.

We consider six agents in a formation system that communicate over
graph $G_b$ as shown in Fig. \ref{fig:topologyb}.

The local dynamics of agents are randomly chosen as
\begin{align}\label{exa:localdyna}
  &A_1=\left[\begin{smallmatrix}
    -8.12 & -2.34 \\
     -2.34 & -4.38 \\
  \end{smallmatrix}\right],
  A_2=\left[\begin{smallmatrix}
    -2.75 & -5.94 \\
     -5.94 & 1.09 \\
  \end{smallmatrix}\right],
  A_3=\left[\begin{smallmatrix}
    5.83 & -3.51 \\
     -3.51 & 3.76 \\
  \end{smallmatrix}\right]\nonumber\\
  &A_4=\left[\begin{smallmatrix}
   9.55 & -3.22 \\
     -3.22 & -11.30 \\
 \end{smallmatrix}\right],
  A_5=\left[\begin{smallmatrix}
   -0.50 & -3.84 \\
     -3.84 & -7.07 \\
 \end{smallmatrix}\right],
  A_6=\left[\begin{smallmatrix}
    3.48 & -2.09 \\
     -2.09 & -5.20 \\
  \end{smallmatrix}\right]
\end{align}
and the nonzero input matrices are
\begin{align}\label{exa:b}
&B_{\bar{11}}= \left[\begin{smallmatrix}
   3.99 & 0 \\
     0 & -3.67 \\
   \end{smallmatrix}\right],
   B_{\bar{12}}=
   \left[\begin{smallmatrix}  -0.15 & 0 \\
     0 & 0.19 \\
   \end{smallmatrix}\right],
B_{\bar{13}}=
   \left[\begin{smallmatrix}3.53 & 0 \\
     0 & -1.33 \\
   \end{smallmatrix}\right], B_{\bar{16}}=
   \left[\begin{smallmatrix}1.06 & 0 \\
     0 & -0.15 \\
   \end{smallmatrix}\right]\nonumber\\
&B_{\bar{22}}=
   \left[\begin{smallmatrix} -3.20 & 0 \\
     0 & 4.41 \\
   \end{smallmatrix}\right],
   B_{\bar{23}}=
   \left[\begin{smallmatrix}2.89 & 0 \\
     0 & -1.54 \\
   \end{smallmatrix}\right],
   B_{\bar{26}}=
  \left[ \begin{smallmatrix} -2.46 & 0 \\
     0 & 0.41 \\
   \end{smallmatrix}\right],\nonumber\\
&B_{\bar{33}}=
   \left[\begin{smallmatrix} -4.69 & 0 \\
     0 & -2.32 \\
   \end{smallmatrix}\right],
B_{\bar{34}}=
   \left[\begin{smallmatrix}1.57 & 0 \\
     0 & -0.63 \\
\end{smallmatrix}\right],
B_{\bar{35}}=
  \left[ \begin{smallmatrix} 1.49 & 0 \\
     0 & -0.94 \\
\end{smallmatrix}\right],\nonumber\\
&B_{\bar{44}}=
   \left[\begin{smallmatrix} -3.22 & 0 \\
     0 & -1.12 \\
\end{smallmatrix}\right],
B_{\bar{45}}=
   \left[\begin{smallmatrix} -1.23 & 0 \\
     0 & 2.98 \\
\end{smallmatrix}\right],
B_{\bar{46}}=
   \left[\begin{smallmatrix} 0.26 & 0 \\
     0 & -1.31 \\
\end{smallmatrix}\right],\nonumber\\
&B_{\bar{55}}=
   \left[\begin{smallmatrix} -3.45 & 0 \\
     0 & 4.49 \\
\end{smallmatrix}\right],
B_{\bar{56}}=
   \left[\begin{smallmatrix} 0.90 & 0 \\
     0 & 0.08 \\
\end{smallmatrix}\right],\nonumber\\
&B_{\bar{66}}=
   \left[\begin{smallmatrix}-1.05 & 0 \\
     0 & -3.65 \\
\end{smallmatrix}\right]
\end{align}
The quadratic matrix is
\begin{align}\label{exa:q}
&Q_{\bar{11}}= \left[\begin{smallmatrix}
    13.50 & 0 \\
     0 & 22.50 \\
   \end{smallmatrix}\right],
   Q_{\bar{12}}=
   \left[\begin{smallmatrix} -1.17 & 0 \\
     0 & -2.63 \\
   \end{smallmatrix}\right],
Q_{\bar{13}}=
   \left[\begin{smallmatrix}-1.74 & 0 \\
     0 & -1.48 \\
   \end{smallmatrix}\right], Q_{\bar{16}}=
   \left[\begin{smallmatrix} -4.21 & 0 \\
     0 & -3.12 \\
   \end{smallmatrix}\right]\nonumber\\
&Q_{\bar{22}}=
   \left[\begin{smallmatrix} 13.50 & 0 \\
     0 & 22.50 \\
   \end{smallmatrix}\right],
   Q_{\bar{23}}=
   \left[\begin{smallmatrix} -5.48 & 0 \\
     0 & -2.59 \\
   \end{smallmatrix}\right],
   Q_{\bar{26}}=
  \left[ \begin{smallmatrix} -4.77 & 0 \\
     0 & -5.31 \\
   \end{smallmatrix}\right],\nonumber\\
&Q_{\bar{33}}=
   \left[\begin{smallmatrix}18 & 0 \\
     0 & 30 \\
   \end{smallmatrix}\right],
Q_{\bar{34}}=
   \left[\begin{smallmatrix} -0.97 & 0 \\
     0 & -0.89 \\
\end{smallmatrix}\right],
Q_{\bar{35}}=
  \left[ \begin{smallmatrix}-1.77 & 0 \\
     0 & -1.24 \\
\end{smallmatrix}\right],\nonumber\\
&Q_{\bar{44}}=
   \left[\begin{smallmatrix}13.50 & 0 \\
     0 & 22.50 \\
\end{smallmatrix}\right],
Q_{\bar{45}}=
   \left[\begin{smallmatrix} -4.84 & 0 \\
     0 & -5.17 \\
\end{smallmatrix}\right],
Q_{\bar{46}}=
   \left[\begin{smallmatrix} -4.01 & 0 \\
     0 & -5.23 \\
\end{smallmatrix}\right],\nonumber\\
&Q_{\bar{55}}=
   \left[\begin{smallmatrix} 18 & 0 \\
     0 & 30 \\
\end{smallmatrix}\right],
Q_{\bar{56}}=
   \left[\begin{smallmatrix} -2.67 & 0 \\
     0 & -2.29 \\
\end{smallmatrix}\right],\nonumber\\
&Q_{\bar{66}}=
   \left[\begin{smallmatrix} 13.50 & 0 \\
     0 & 22.50 \\
\end{smallmatrix}\right]
\end{align}
and all the other blocks are zeros. Using Algorithm \ref{alg1}, the
subgradient matrix converges to zero after approximately 63 times of
iterations. The interaction parameters in matrix $A'$ are
\begin{align}\label{eq:exmAde} &A'_{\bar{12}}=
   \left[\begin{smallmatrix} 4.79 & 3.60 \\
     -1.12 & -0.96 \\
   \end{smallmatrix}\right],
A'_{\bar{13}}=
   \left[\begin{smallmatrix}-3.30 & 3.43 \\
     6.41 & 2.21 \\
   \end{smallmatrix}\right],
   A'_{\bar{16}}=
   \left[\begin{smallmatrix} 2.68 & -2.72 \\
     0.34 & -0.57 \\
   \end{smallmatrix}\right],\nonumber\\
&   A'_{\bar{23}}=
   \left[\begin{smallmatrix}  -3.23 & 1.97 \\
     10.33 & 0.61 \\
   \end{smallmatrix}\right],
   A'_{\bar{26}}=
  \left[ \begin{smallmatrix}2.31 & -1.82 \\
     -0.96 & -4.78 \\
   \end{smallmatrix}\right],\nonumber\\
& A'_{\bar{34}}=
   \left[\begin{smallmatrix} -3.86 & -0.23 \\
     -0.10 & 1.07 \\
\end{smallmatrix}\right],
A'_{\bar{35}}=
  \left[ \begin{smallmatrix} -2.14 & 0.69 \\
     -1.99 & 0.56 \\
\end{smallmatrix}\right],\nonumber\\
&A'_{\bar{45}}=
   \left[\begin{smallmatrix} 3.65 & 0.63 \\
     -0.38 & -3.88 \\
\end{smallmatrix}\right],
A'_{\bar{46}}=
   \left[\begin{smallmatrix}  -0.38 & -2.40 \\
     -0.35 & -2.73 \\
\end{smallmatrix}\right],\nonumber\\
&A'_{\bar{56}}=
   \left[\begin{smallmatrix} -0.54 & 3.93 \\
     0.14 & 3.23 \\
\end{smallmatrix}\right],
\end{align}
and all the other blocks are zeros. According to the structure of
matrix $A'$, the underlying graph is exactly the one in Fig.
\ref{fig:topologyb}, and the  cost function achieves the minimum of
$J^*=3.39$ for the worst case. We tested 50 samples of randomly
selected $A'$ under the structure constraints. Greater cost values
are observed on all of those samples. The data are recorded in Fig.
\ref{fig:six_costvalue}. The dashed line at the bottom is the value
of $\|X^*\|=3.39$ and all the blue stars are different worst-case
cost values $\bar{J}$ when $A'\in T_G$ changes.  This validates the
effectiveness of the subgradient algorithm.
\begin{figure}
\begin{center}
\includegraphics[scale=0.4]{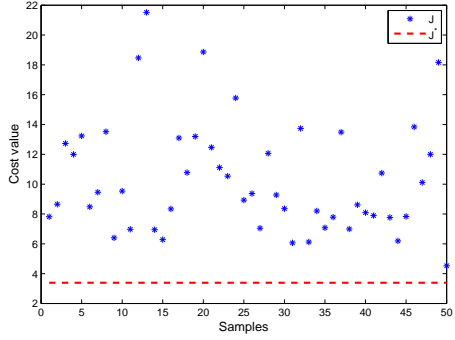}
\caption{Worst-case cost values of six-agent formations. The bottom
dashed line is the optimized cost value
$J^*$}\label{fig:six_costvalue}
\end{center}
\end{figure}
Fig. \ref{fig:trjc_six_opt} is an example of formations of the six
agents under the matched pair $(A_0+A',Q)$. Snapshots during the
process are shown in Fig. \ref{fig:trjc_six_snapshot}.
\begin{figure}
\begin{center}
\includegraphics[scale=0.4]{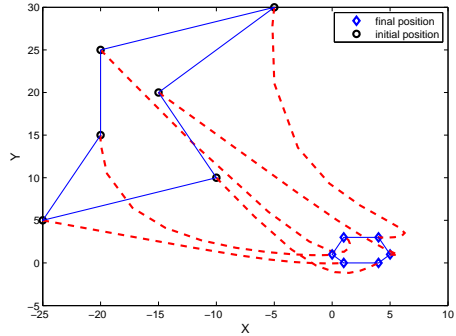}
\caption{Formations of six agents}\label{fig:trjc_six_opt}
\end{center}
\end{figure}

\begin{figure}
 \subfigure[$t=0s$]{
\begin{minipage}[b]{0.23\linewidth}
\centering
\includegraphics[scale=0.2]{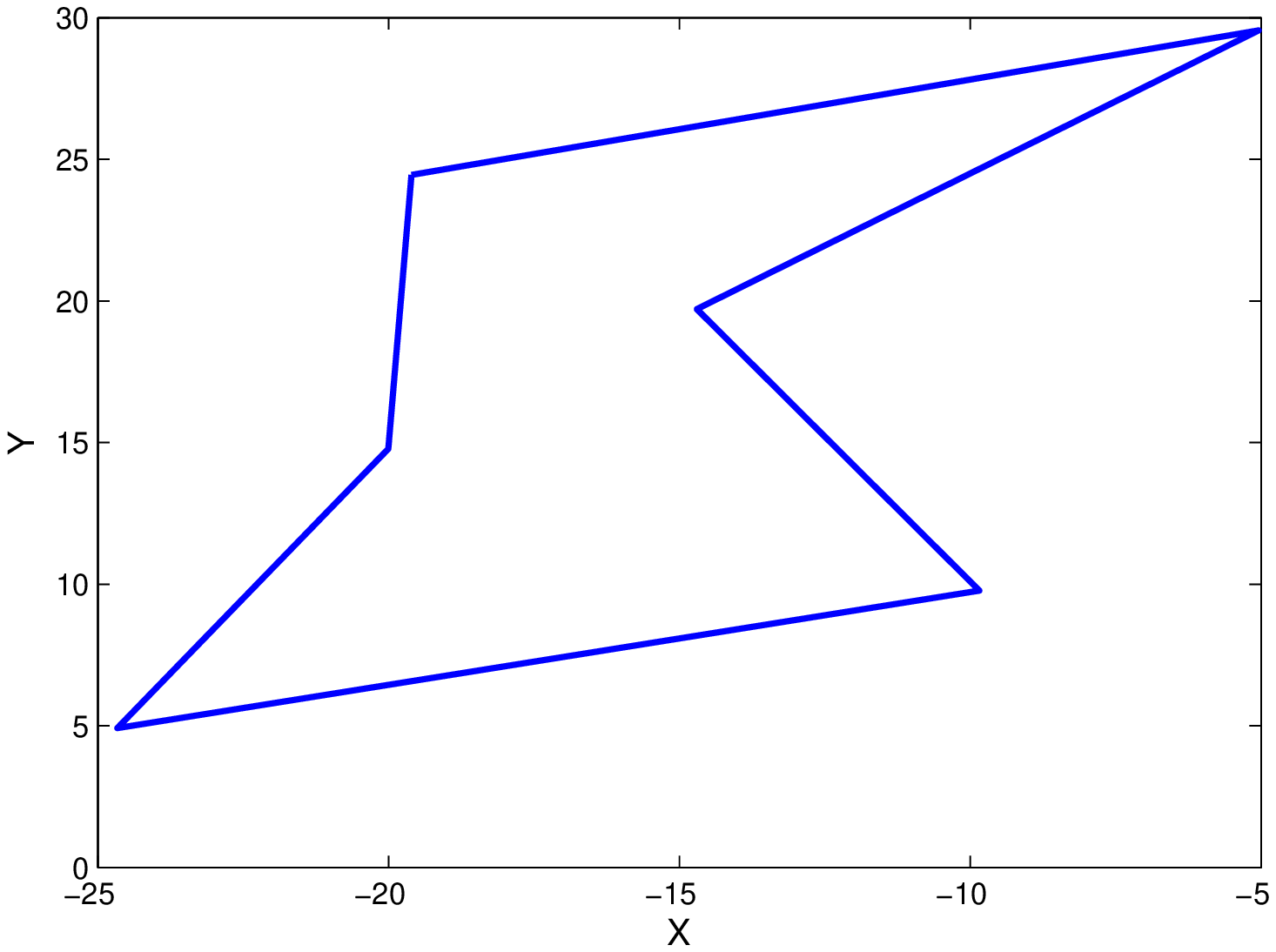}
\end{minipage}}
\subfigure[$t=0.22s$]{
\begin{minipage}[b]{0.23\linewidth}
\centering
\includegraphics[scale=0.2]{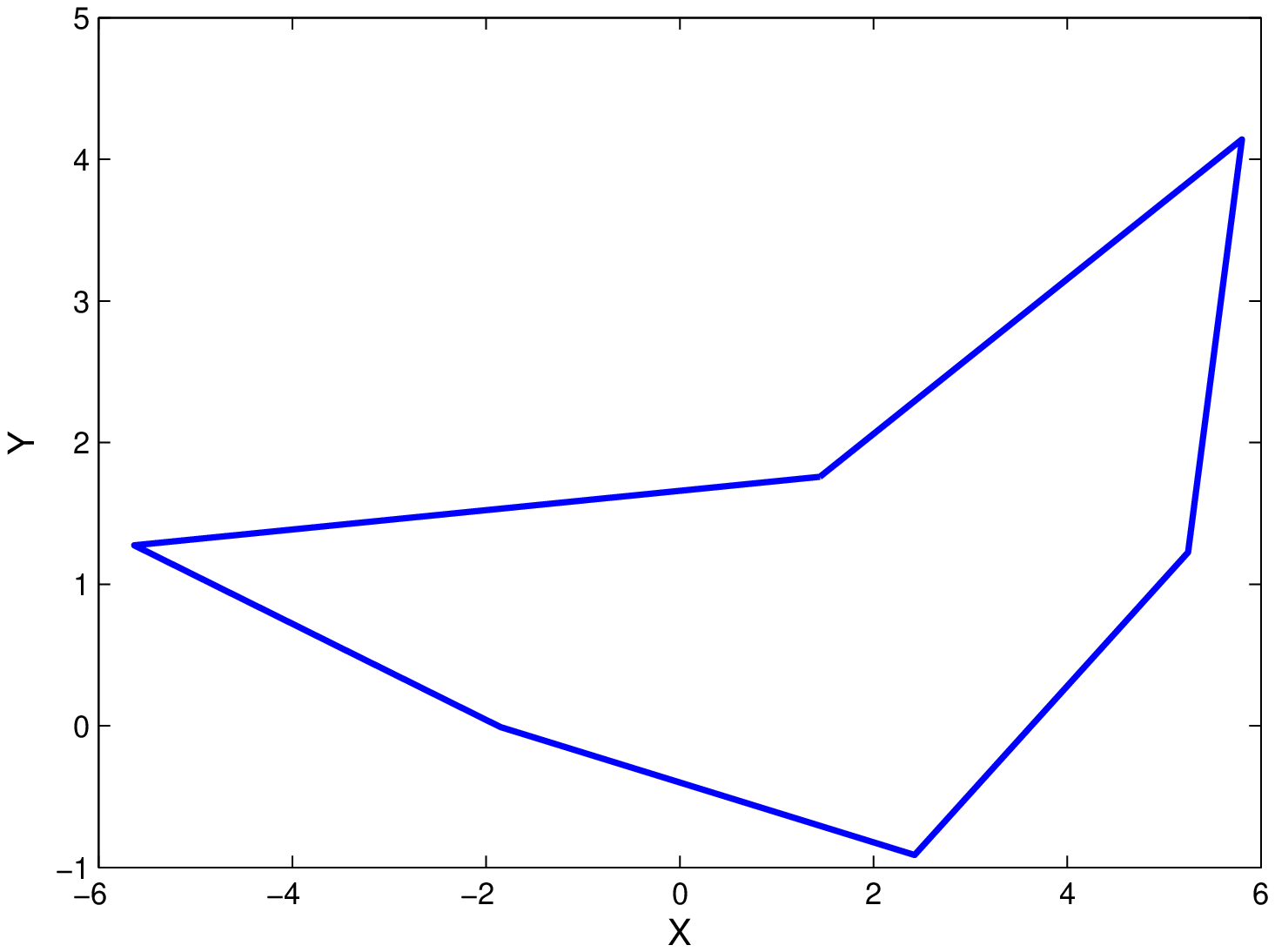}
\end{minipage}}
\subfigure[$t=0.66s$]{
\begin{minipage}[b]{0.23\linewidth}
\centering
\includegraphics[scale=0.2]{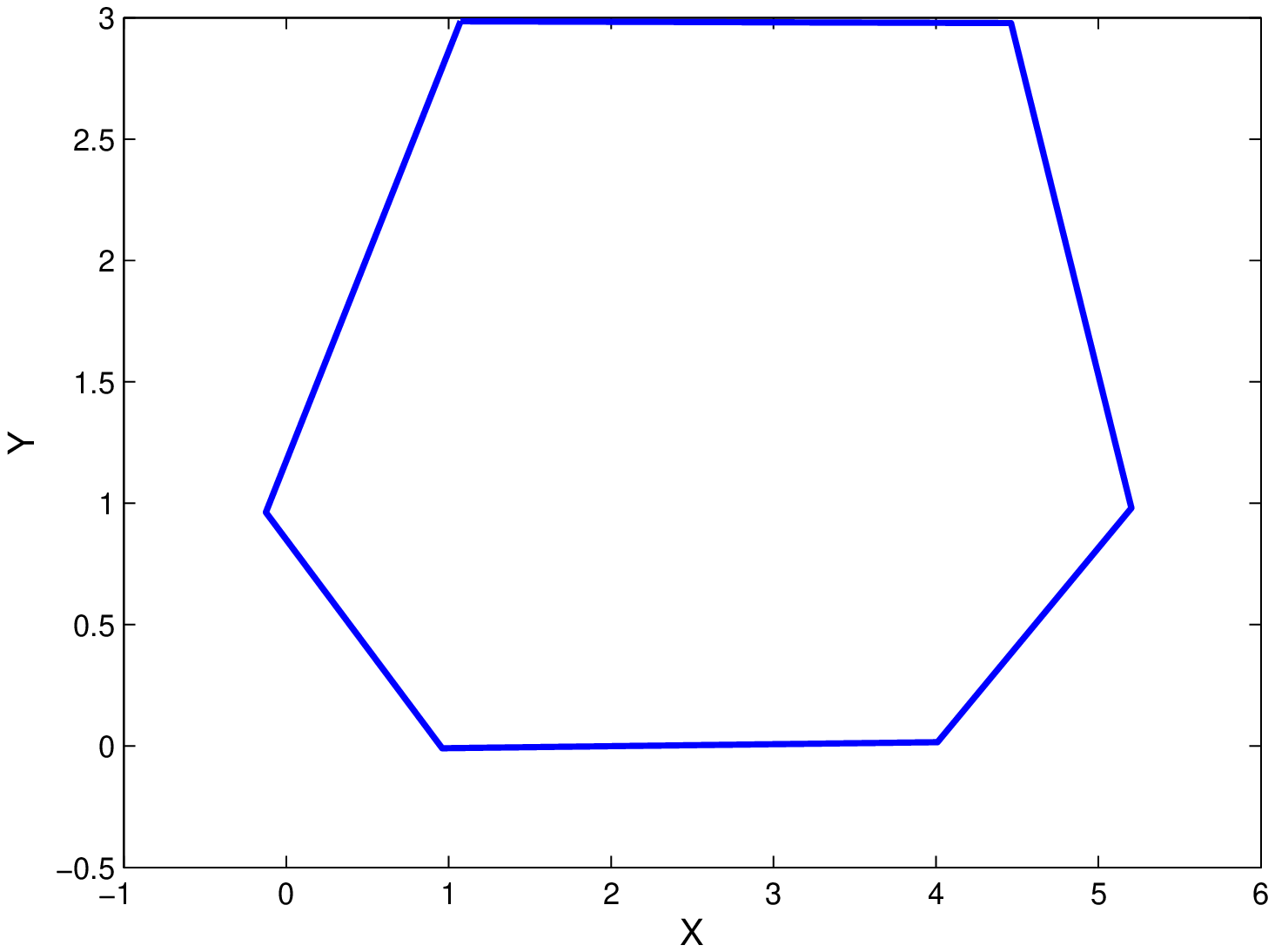}
\end{minipage}}
\subfigure[$t=1.23s$]{
\begin{minipage}[b]{0.23\linewidth}
\centering
\includegraphics[scale=0.2]{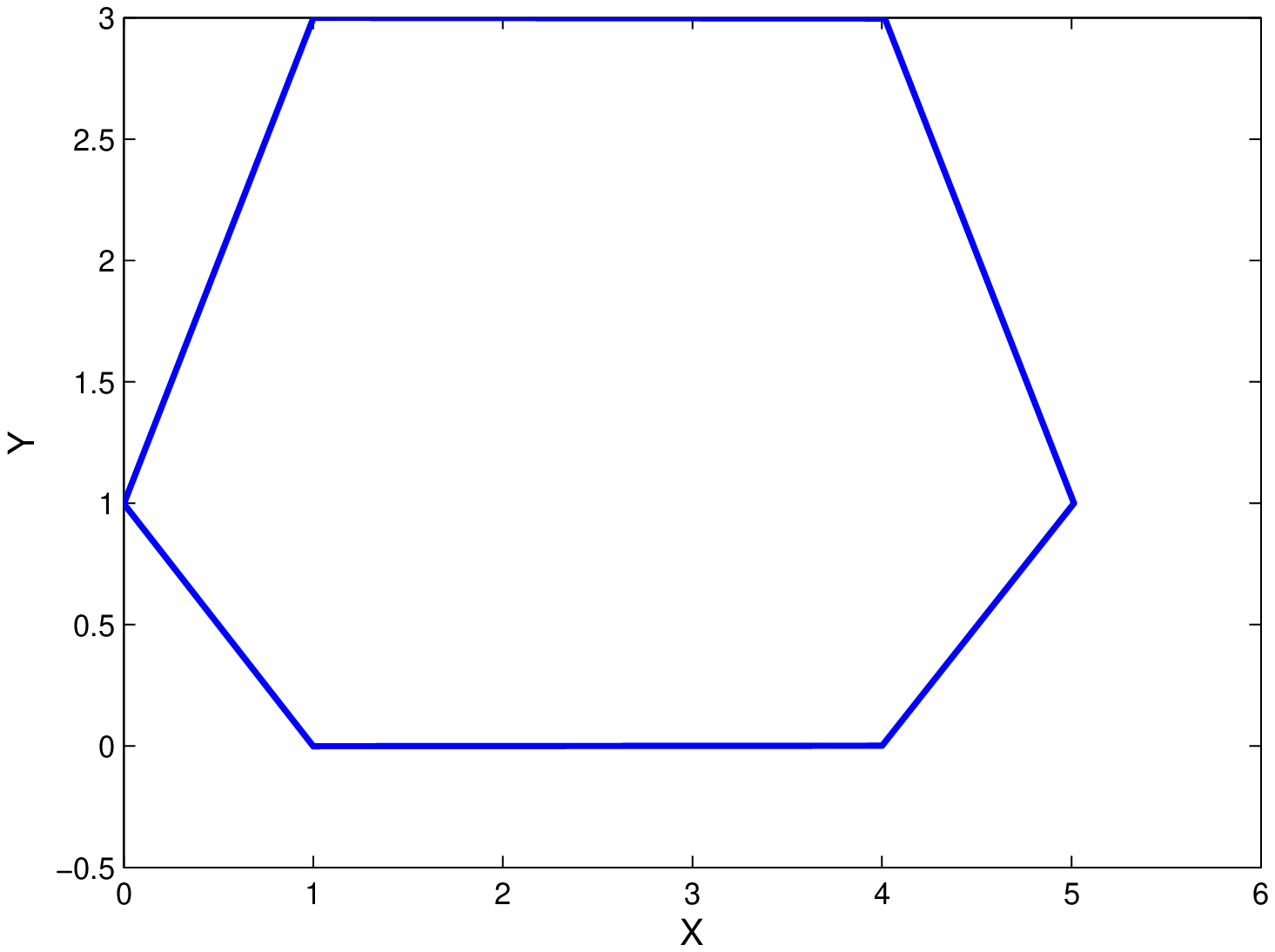}
\end{minipage}}
\caption{Snapshots of the six agents during
attainment}\label{fig:trjc_six_snapshot}
\end{figure}

For comparison, we consider a formation system where agents are
dynamically isolated. The local dynamics, the input matrix and the
cost matrix are congruent to the ones in \eqref{exa:localdyna},
\eqref{exa:b} and \eqref{exa:q} respectively. The optimal
controllers are calculated by the ARE \eqref{eq:areB}. For the worst
case, the cost function has a minimal value of $\bar{J}=4.53$, as
shown in sample number 50 in Fig. \ref{fig:six_costvalue}. We
compare the performance of a system with isolated agents and a
system with optimized couplings \eqref{eq:exmAde}. For both the
situations, agents are initialized at the worst case, i.e.,
$p_0=-3v$ with $v$ being the orthogonal eigenvector of
$\bar{\lambda}(A_0)$ and $\bar{\lambda}(A_0+A')$ respectively. The
two plots in Fig. \ref{fig:trjc_six_compare} show their trajectories
during the stabilization processes. The formation errors for system
under the matched pair $(A_0+A',Q)$ and the pair $(A_0,Q)$ are shown
in Fig. \ref{fig:fe_struc_A}. For the two situations, the snapshots
of the geometries at time slot $t=0.21sec$ are shown in Fig.
\ref{fig:structed_snap}. According to Fig. \ref{fig:fe_and_snap},
the geometry of agents under the matched pair is more similar to the
desired geometry at the same time slot. Formation systems with
isolated agents have relatively greater value of cost function, and
thus its relative formation errors are consistently greater than the
one with agents interacted over optimal parameters.

\begin{figure}
 \subfigure[Optimized couplings]{\label{fig:trjc_six_compare_opt}
\begin{minipage}[b]{0.45\linewidth}
\centering
\includegraphics[scale=0.25]{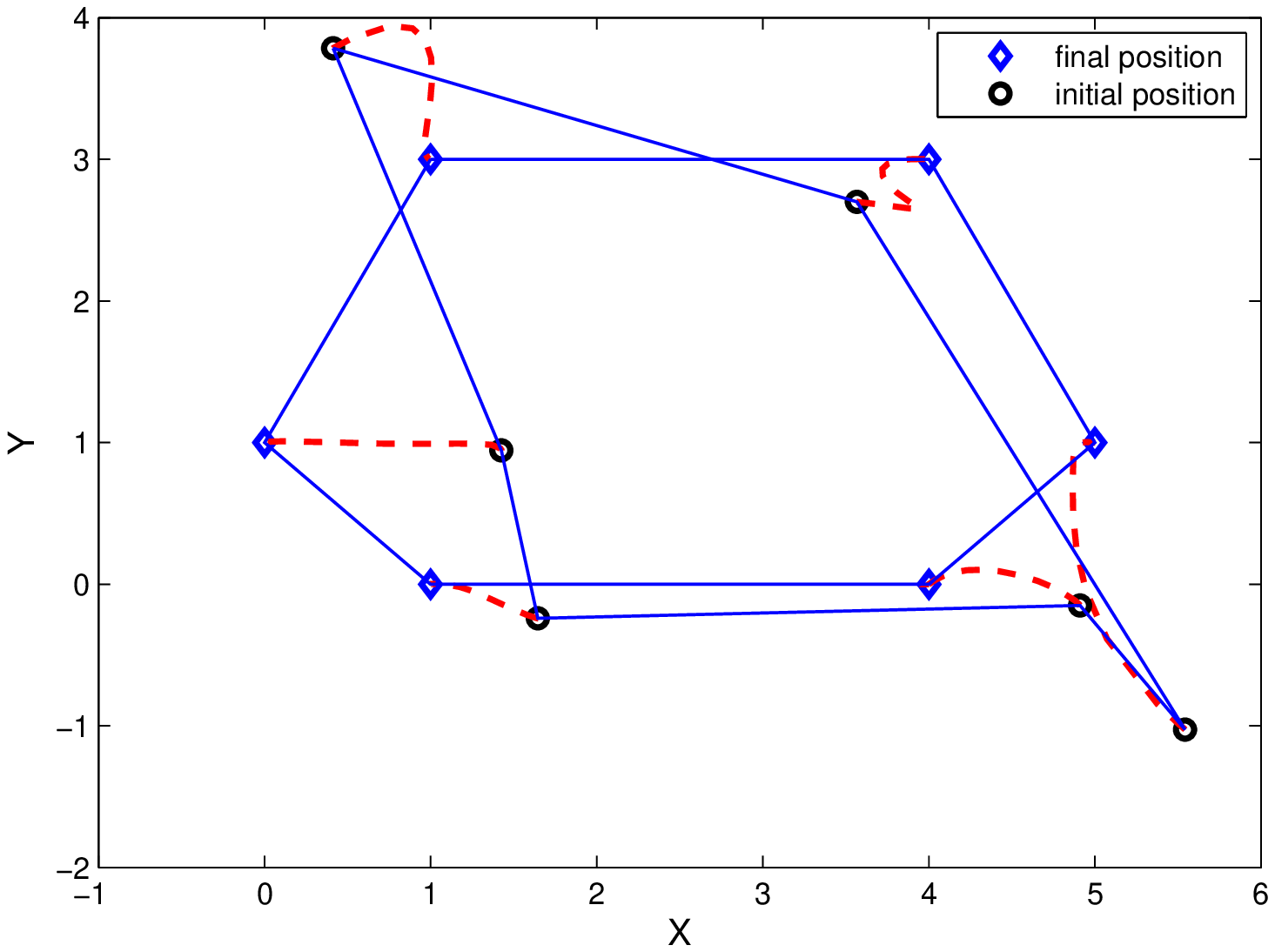}
\end{minipage}}
\subfigure[Isolated agents]{\label{fig:trjc_six_compare_de}
\begin{minipage}[b]{0.45\linewidth}
\centering
\includegraphics[scale=0.25]{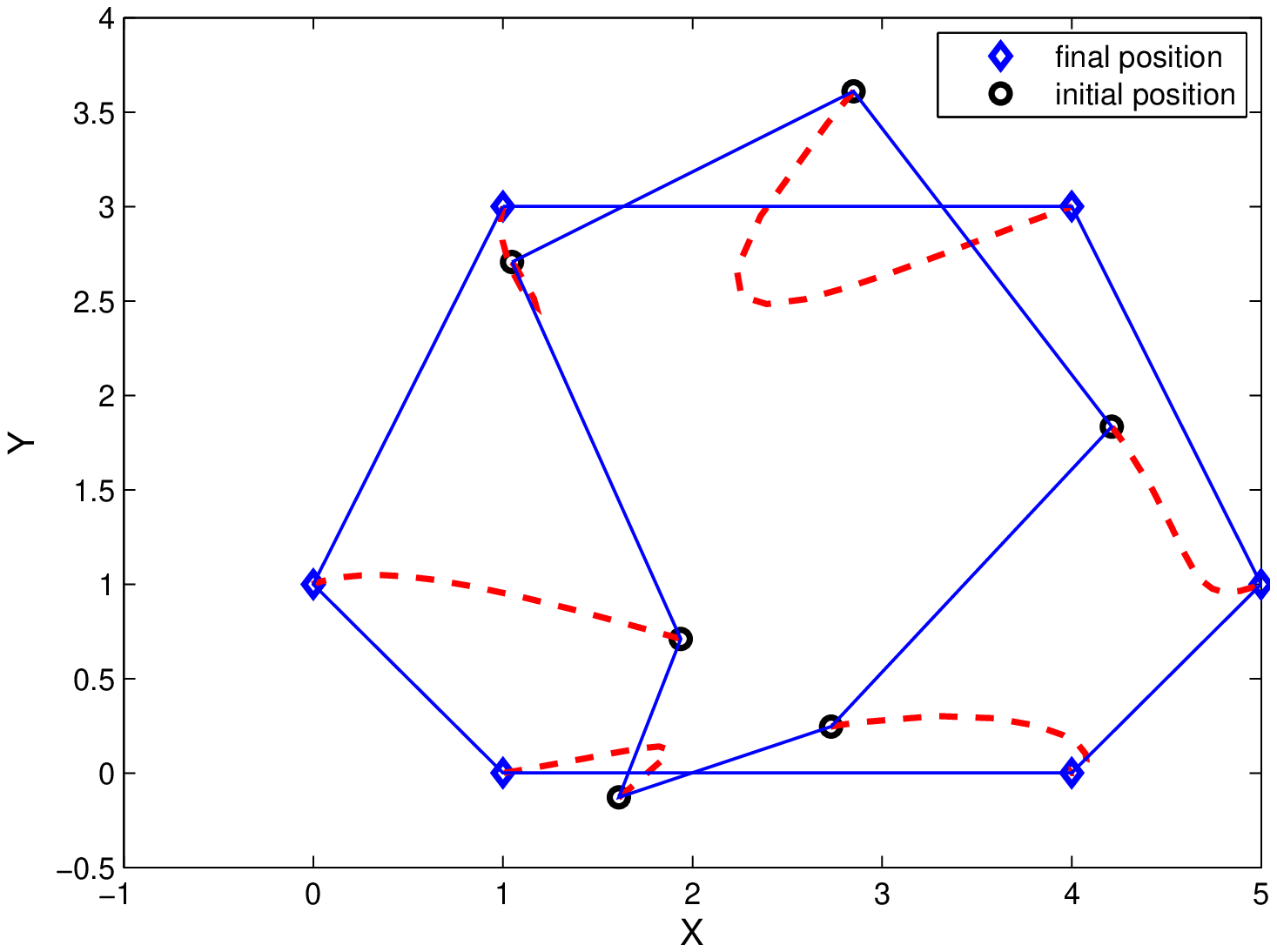}
\end{minipage}}
\caption{Trajectories of six agents being initialized at their worst
cases.}\label{fig:trjc_six_compare}
\end{figure}

\begin{figure}
 \subfigure[Formation errors]{\label{fig:fe_struc_A}
\begin{minipage}[b]{0.45\linewidth}
\centering
\includegraphics[scale=0.25]{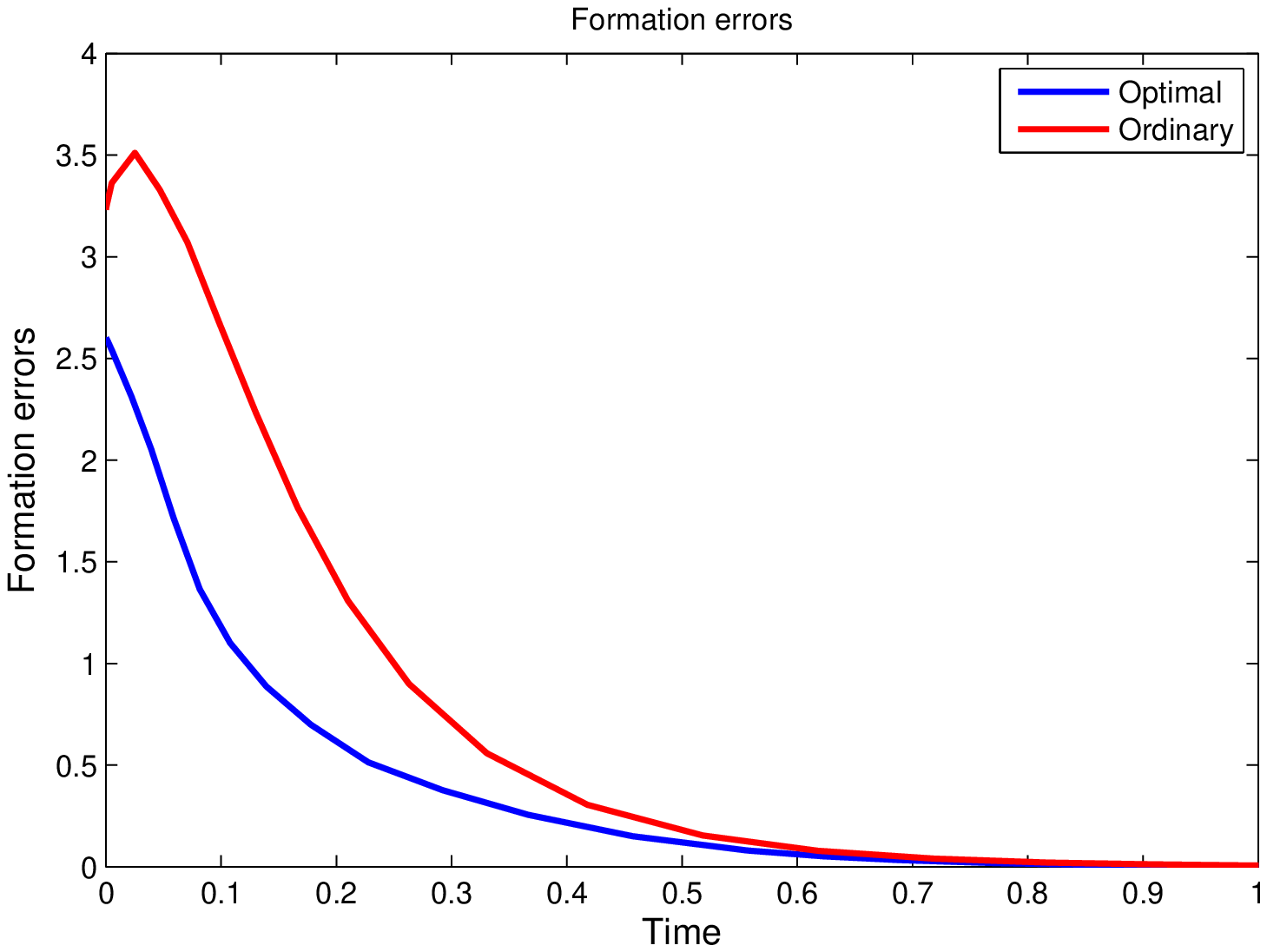}
\end{minipage}}
\subfigure[Snapshots at $t=0.21sec$]{\label{fig:structed_snap}
\begin{minipage}[b]{0.45\linewidth}
\centering
\includegraphics[scale=0.25]{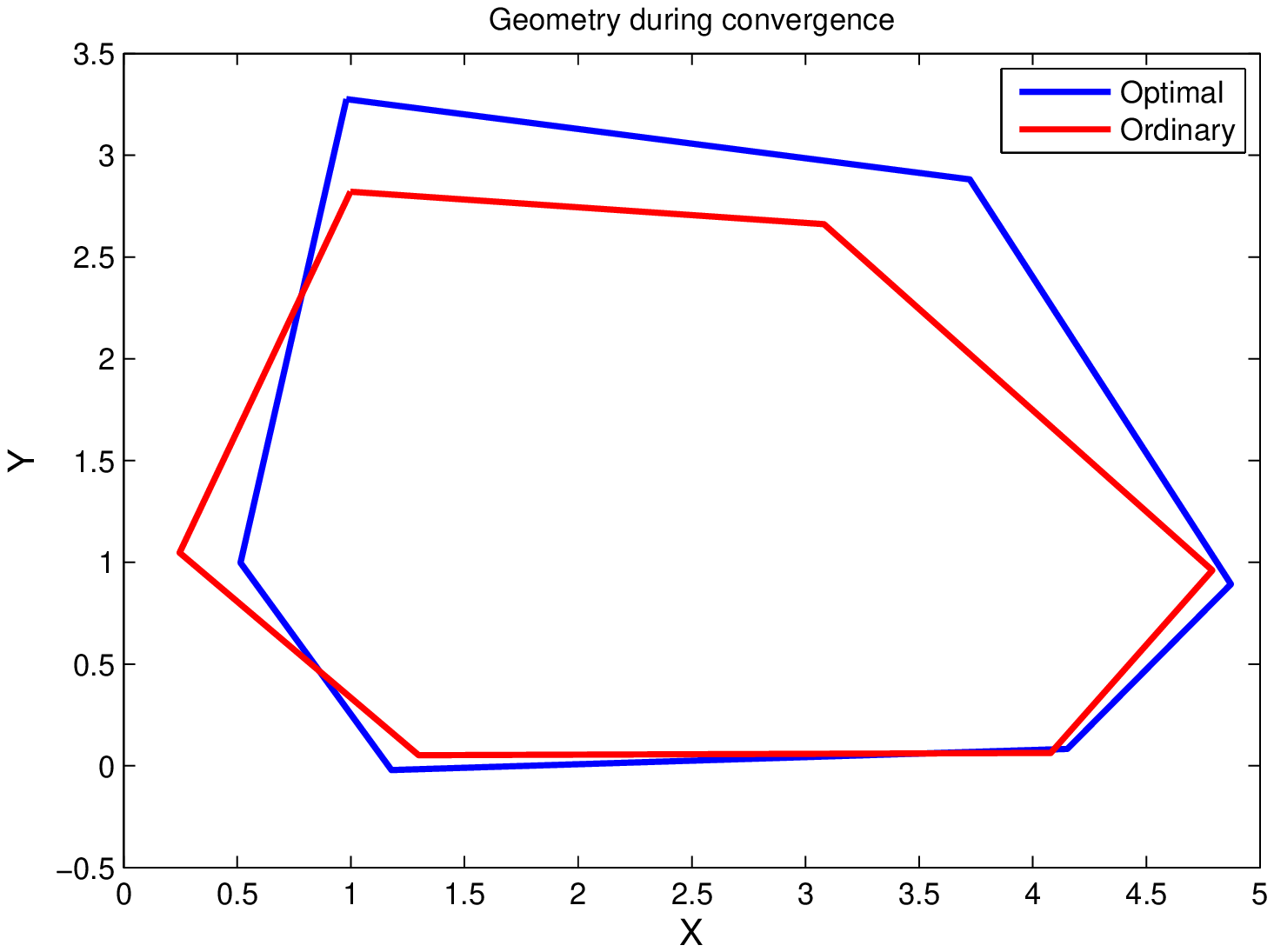}
\end{minipage}}
\caption{Formation errors and the snapshots of systems with isolated
agents and optimal coupled agents}\label{fig:fe_and_snap}
\end{figure}

%
%

\section{Conclusions}
In this paper, we proposed an optimal LQR control strategy for a
group of agents to maintain desired geometries while moving towards
the destination. Upon the three kinds of cooperative performance
that were characterized by the cost matrix, it was proved that,
compared to system with agents communicated only through the control
channel, the upper bound of the minimum cost value could be further
reduced by adjusting the interaction parameters between the pairs of
neighboring agents and by the LQR controllers. Distributed
controllers that inherited the desired underlying graph were
developed for the set of homogeneous agents. When agents were
heterogenous ones, the optimization problem was relaxed so as to
avoid the nonsmoothness and nonlinearity in the constraints. Under
the constraint of a sparse underlying graph, parameters among agents
were generated by a reliable iterative algorithm. Numerical examples
further illustrated the relationship between the cost function and
the cooperative performance we considered.

\bibliographystyle{siam}
\bibliography{C:/papers/odds/odds}
\end{document}